\documentclass[a4paper,runningheads]{llncs}

\usepackage{xspace,amsmath,amsfonts,amssymb,bbm,ifthen,stmaryrd,mathtools}
\usepackage[OT4]{fontenc}
\usepackage{enumerate,subfigure,url}

\usepackage[caps]{complexity}
\newclass{\EXPTIME}{Exptime}
\renewclass{\NP}{NP}
\renewclass{\coNP}{co-NP}

\usepackage{tikz}
\usetikzlibrary{arrows,automata,shapes}

\def\longrightarrowtriangle{\relbar\joinrel\rightarrowtriangle}

\newcommand{\conj}{\varowedge}

\newcommand*\may{\textit{may}\xspace}
\newcommand*\must{\textit{must}\xspace}
\newcommand*\Must{\textit{Must}\xspace}
\newcommand*\mayto[1][]{%
  \ensuremath{\mathrel{\overset{%
     \raisebox{3pt}{$\scriptscriptstyle%
        #1%
      $}%
    }{\smash{\dashrightarrow}}}}}
\newcommand*\mustto[1][]{%
  \ensuremath{\mathrel{\overset{%
     \raisebox{3pt}{$\scriptscriptstyle%
        #1%
      $}%
    }{\smash{\longrightarrow}}}}}
\newcommand*\tto{\mustto}
\newcommand*\gameto[1][]{\ensuremath{\overset{#1}{\longrightarrowtriangle}}}
\newcommand*\mmayto{\ensuremath{\mathord{\mayto}}}
\newcommand*\mmustto{\ensuremath{\mathord{\mustto}}}
\newcommand*\mtto{\mmustto}
\newcommand*\mgameto{\ensuremath{\mathord{\gameto}}}
\newcommand*\Int{\mathbbm{Z}}
\newcommand*\Nat{\mathbbm{N}}
\newcommand*\Real{\ensuremath{\mathbbm{R}}}
\newcommand*\Realnn{\ensuremath{\Real_{ \ge 0}}}
\newcommand*\I{\ensuremath{\mathbbm{I}}}

\renewcommand*\K{\textup{\textsf{Spec}}}
\newcommand*\Imp[1]{\textup{\textsf{Imp}}}

\DeclareMathOperator{\id}{id}
\DeclareMathOperator{\pre}{pre}
\newcommand*\bigmid{\mathrel{\big|}}

\renewcommand*\epsilon{\varepsilon}
\newcommand*\powerset[1]{\mathcal P(#1)}

\newcommand*\tgt{\textit{tgt}} 
\newcommand*\weight{\textit{wt}} 

\newcommand*\ie{\textit{i.e.}\xspace}
\newcommand*\cf{\textit{cf.}\xspace}
\newcommand*\eg{\textit{e.g.}\xspace}

\newcommand*\sem[1]{\llbracket #1 \rrbracket}

\newcommand*\idist{d_{ \Imp \K}} 
\newcommand*\kdist{d_\K} 
\newcommand*\bdist{d} 
\newcommand*\mdist{d_m} 
\newcommand*\tdist{d_t} 

\newcommand*\true{t\!t}
\newcommand*\false{f\!\!f}

\title{Weighted Modal Transition Systems%
  \thanks{This paper is based on the conference
    contribution~\cite{DBLP:conf/mfcs/BauerFJLLT11} which was presented
    at the 36th International Symposium on Mathematical Foundations of
    Computer Science, MFCS 2011, Warszawa, Poland.}}

\author{Sebastian S.~Bauer\inst1 \and Uli Fahrenberg\inst2 \and Line
  Juhl\inst3 \and Kim G.~Larsen\inst3 \and Axel Legay\inst2 \and Claus
  Thrane\inst3}

\authorrunning{S.S.~Bauer, U.~Fahrenberg, L.~Juhl, K.G.~Larsen,
  A.~Legay, C.~Thrane}

\institute{Ludwig-Maximilians-Universit{\"a}t M{\"u}nchen, Germany \and
  Irisa/INRIA Rennes, France \and Aalborg University, Denmark}

\begin{document}

\maketitle

\begin{abstract}
  Specification theories as a tool in model-driven development processes
  of component-based software systems have recently attracted a
  considerable attention.  Current specification theories are however
  qualitative in nature, and therefore fragile in the sense that the
  inevitable approximation of systems by models, combined with the
  fundamental unpredictability of hardware platforms, makes it difficult
  to transfer conclusions about the behavior, based on models, to the
  actual system.  Hence this approach is arguably unsuited for modern
  software systems.  We propose here the first specification theory
  which allows to capture quantitative aspects during the refinement and
  implementation process, thus leveraging the problems of the
  qualitative setting.

  Our proposed quantitative specification framework uses weighted modal
  transition systems as a formal model of specifications. These are
  labeled transition systems with the additional feature that they can
  model optional behavior which may or may not be implemented by the
  system.  Satisfaction and refinement is lifted from the well-known
  qualitative to our quantitative setting, by introducing a notion of
  distances between weighted modal transition systems. We show that
  quantitative versions of parallel composition as well as quotient (the
  dual to parallel composition) inherit the properties from the Boolean
  setting.
\end{abstract}

\section{Introduction}

One of the major current challenges to rigorous design of software
systems is that these systems are becoming increasingly complex and
difficult to reason about~\cite{Sifakis11}.  As an example, an
integrated communication system in a modern airplane can have more than
$10^{ 900}$ distinct states~\cite{DBLP:conf/forte/BasuBBCDL10}, and
state-of-the-art tools offer no possibility to reason about, and model
check, the system as a whole.

One promising approach to overcome such problems is the one of
\emph{compositional and incremental design}.  Here the reasoning is done
as much as possible at higher \emph{specification} levels rather than at
\emph{implementations}; partial specifications are proven correct and
then composed and refined until one arrives at an implementation model.
Practice has shown that this is indeed a viable
approach~\cite{COMBEST,SPEEDS}.

Specifications of system requirements are high-level finite abstractions
of possibly infinite sets of implementations. A model of a system is
considered an implementation of a given specification if the behavior
defined by the implementation is implied by the description provided by
the specification. 

Any practical specification formalism comes equipped with a number of
operations which allow compositional and incremental reasoning.  The
first of these is a \emph{refinement} relation which allows to
successively distill specifications into more detailed ones and
eventually into implementations.  In an implementation, all optional
behavior defined in the specification has been decided upon in
compliance with the specification.

Also needed is an operation of \emph{logical conjunction} which allows
to combine specifications so that the systems which refine the
conjunction of two specifications are precisely the ones which satisfy
both partial specifications.  Refinement and conjunction together allow
for incremental reasoning as specifications are successively refined and
composed.

For compositional reasoning, one needs an operation of \emph{structural
  composition} which allows to infer specifications from
sub-specifications of independent requirements, mimicking at the
implementation level \eg~the interaction of components in a distributed
system. A partial inverse of this operation is given by the
\emph{quotient} operation which allows to synthesize a specification of
the missing components from an overall specification and an
implementation which realizes a part of the overall specification.

Over the years, there have been a series of advances on specification
theories~\cite{AlfaroH95,ChakrabartiAHM02,DBLP:conf/hybrid/DavidLLNW10,Delahaye10Thesis,Lynch-tuttle88,Nyman08Thesis,Thrane11Thesis}.
The predominant approaches are based on modal logics and process
algebras but have the drawback that they cannot naturally embed both
logical and structural composition within the same
formalism~\cite{Larsen89}.  Hence such formalisms do not permit to
reason incrementally through refinement.

In order to leverage these problems, the concept of \emph{modal
  transition systems} was introduced~\cite{Larsen89}. In short, modal
transition systems are labeled transition systems equipped with two
types of transitions: \must~transitions which are mandatory for any
implementation, and \may~transitions which are optional for
implementations.  It is well established that modal transition systems
match all the requirements of a reasonable specification theory (see
also~\cite{DBLP:journals/entcs/Raclet08} for motivation), and much
progress has been made using modal specifications, see
\eg~\cite{AntonikHLNW08} for an overview.  Also, practical experience
shows that the formalism is expressive enough to handle complex
industrial problems~\cite{COMBEST,SPEEDS}.

As an example, consider the modal transition system shown in
Figure~\ref{fig:mtsintro} which models the requirements of a simple
email system in which emails are first received and then delivered.
Before delivering the email, the system may check or process the email,
\eg~for en- or decryption, filtering of spam emails, or generating
automatic answers using has an auto-reply feature (see
also~\cite{DBLP:conf/fiw/Hall00}).  \Must transitions, representing
obligatory behavior, are drawn as solid arrows, whereas \may
transitions, modeling optional behavior, are shown as dashed arrows;
hence any implementation of this email system specification must be able
to receive and deliver email, and it may also be able to check arriving
email before delivering it.  No other behavior is allowed.

\begin{figure}[tp]
  \centering
  \begin{tikzpicture}[->,>=stealth',shorten >=1pt,auto,node
    distance=2.0cm,initial text=,scale=0.9,transform shape]
    \tikzstyle{every node}=[font=\small] \tikzstyle{every
      state}=[fill=white,shape=circle,inner sep=.5mm,minimum size=6mm]
    \node[state,initial] (s0) at (0,0) {};
    \node[state] (s1) at (3,0) {};
    \node[state] (s2) at (6,0) {};
    \path (s0) edge [solid] node [above,sloped] {receive} (s1);
    \path (s1) edge [solid,bend right] node [above,sloped] {deliver} (s0);
    \path (s1) edge [densely dashed] node [above] {check} (s2);
    \path (s2) edge [solid,bend left=25] node [below,sloped] {deliver} (s0);
  \end{tikzpicture} 
  \caption{Modal transition system modeling a simple email system, with
    an optional behavior: Once an email is received it may \eg~be
    scanned for containing viruses, or automatically decrypted, before
    it is delivered to the receiver.}
  \label{fig:mtsintro}
\end{figure}
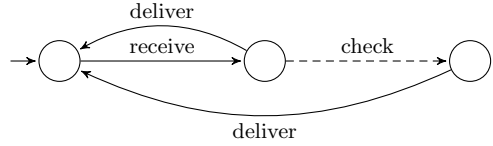

Implementations can also be represented within the modal transition
system formalism, simply as specifications without \may transitions.
Hence any implementation choice has been resolved, and implementations
are plain labeled transition systems.  Formally, for
a labeled transition system to be an implementation of a given
specification,
we require that the states of the two objects are related by a
refinement relation with the property that all behavior required (\must)
by the specification has been implemented, and that any implementation
behavior was permitted (\may) in the specification.
Figure~\ref{fig:intro_impl} shows an implementation of our email
specification with two different
checks, leading to distinct processing states.
Note that a simple system without any check at all, hence only able to
receive and deliver email, is also an implementation of the
specification.

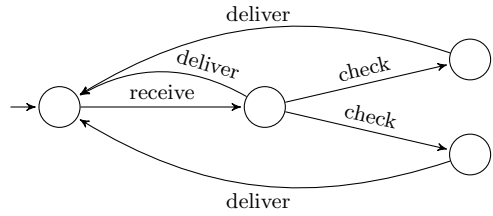
\begin{figure}[tp]
  \centering
  \begin{tikzpicture}[->,>=stealth',shorten >=1pt,auto,node
    distance=2.0cm,initial text=,scale=0.9,transform shape]
    \tikzstyle{every node}=[font=\small] \tikzstyle{every
      state}=[fill=white,shape=circle,inner sep=.5mm,minimum size=6mm]
    \node[state,initial] (s0) at (0,0) {};
    \node[state] (s1) at (3,0) {};
    \node[state] (s2) at (6,-.7) {};
    \node[state] (s3) at (6,.7) {};
    \path (s0) edge [solid] node [above,sloped] {receive} (s1);
    \path (s1) edge [solid,bend right] node [near start,above,sloped]
    {deliver} (s0);
    \path (s1) edge [solid] node [above,sloped] {check} (s2);
    \path (s1) edge [solid] node [above,sloped] {check} (s3);
    \path (s2) edge [solid,bend left=25] node [below] {deliver} (s0);
    \path (s3) edge [solid,bend right=25] node [above] {deliver} (s0);
  \end{tikzpicture} 
  \caption{An implementation of the simple email system in
    Figure~\ref{fig:mtsintro} in which we explicitly model two distinct
    types of email pre-processing.}
  \label{fig:intro_impl}
\end{figure}

\medskip
Motivated by applications to embedded, real-time and hybrid systems, the
modal transition system framework has recently been extended in order to
reason about \emph{quantitative}
aspects~\cite{journals/mscs/BauerJLLS11,journals/jlp/JuhlLS11}.  With
these applications in mind, it is necessary not only to be able to
\emph{specify} quantitative aspects of systems, but also to formalize
successive \emph{refinement} of quantities.
To illustrate this extension, consider again the modal transition system
of Figure~\ref{fig:mtsintro}, but this time with quantities, see
Figure~\ref{fig:mtsquantities}: Every transition label is extended by
integer intervals modeling upper and lower bounds on time required for
performing the corresponding actions. For instance, the reception of a
new email (action \emph{receive}) must take between one and three time
units, the checking of the email (action \emph{check}) is allowed to
take up to five time units.

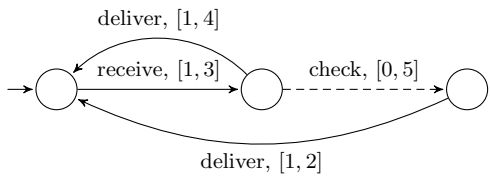
\begin{figure}[tp]
  \centering
  \begin{tikzpicture}[->,>=stealth',shorten >=1pt,auto,node
    distance=2.0cm,initial text=,scale=0.9,transform shape]
    \tikzstyle{every node}=[font=\small] \tikzstyle{every
      state}=[fill=white,shape=circle,inner sep=.5mm,minimum size=6mm]
    \node[state,initial] (s0) at (0,0) {};
    \node[state] (s1) at (3,0) {};
    \node[state] (s2) at (6,0) {};
    \path (s0) edge [solid] node [above,sloped] {receive, $[1,3]$} (s1);
    \path (s1) edge [solid,bend right=45] node [above,sloped] {deliver,
      $[1,4]$} (s0); 
    \path (s1) edge [densely dashed] node [above] {check, $[0,5]$} (s2);
    \path (s2) edge [solid,bend left=25] node [below,sloped] {deliver,
      $[1,2]$} (s0); 
  \end{tikzpicture} 
  \caption{Specification of a simple email system, similar to
    Figure~\ref{fig:mtsintro}, but extended by integer intervals
    modeling time units for performing the corresponding actions.}
  \label{fig:mtsquantities}
\end{figure}

In this quantitative setting, there is a problem with using a
\emph{Boolean} notion of refinement (as is done
in~\cite{journals/mscs/BauerJLLS11,journals/jlp/JuhlLS11}): If one only
can decide \emph{whether or not} an implementation refines a
specification, then the quantitative aspects get lost in the refinement
process.  As an example, consider the email system implementations in
Figure~\ref{fig:introimpl}.  Implementation~(a) does not refine the
specification, as there is an error in the discrete structure of
actions: after receiving an email, the system can check it indefinitely
without ever delivering it.  Also implementations~(b) and~(c) do not
refine the specification: (b) takes too long to receive email, (c) does
not deliver email fast enough after checking it.  Implementation~(d) on
the other hand is a perfect refinement of the specification.

Intuitively however, implementations~(b) and~(c) conform much better to
the specification than implementation~(a) in Figure~\ref{fig:introimpl}:
there are no discrepancies in the discrete structure, only the weights
are off by $1$.  Additionally, the quantitative error in
implementation~(c) occurs later than the one in~(b).  Hence one may want
to say that implementation~(d) is in perfect refinement of the
specification, (c) is slightly off, (b) is a bit more problematic,
whereas implementation~(a) is completely unacceptable.  A Boolean notion
of refinement does not allow to make such distinctions between different
negative answers.

To sum up, a Boolean notion of refinement is too \emph{fragile} for
quantitative formalisms.  Minor
and major modifications in the implementation cannot be distinguished,
as both of them may reverse the Boolean answer.  As observed 
in~\cite{DBLP:journals/tse/AlfaroFS09}, this view is obsolete; engineers
need quantitative notions on how modified implementations differ.
The introduction of such a quantitative notion of refinement, and its
consequences for the specification theory, are the subject of this paper.

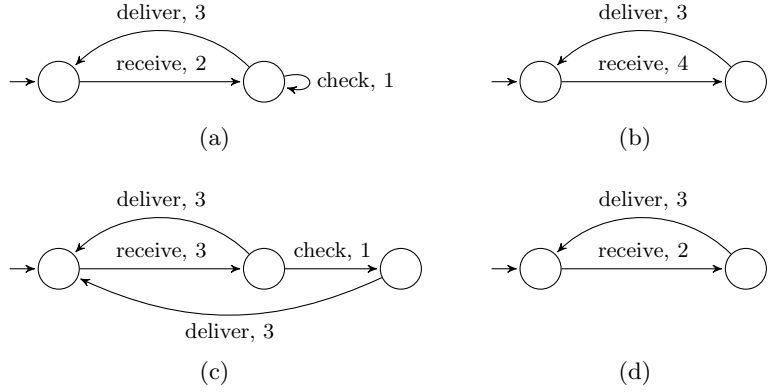
\begin{figure}[tp]
  \centering
\subfigure[]{
    \label{fig:introimpl:1}
  \begin{tikzpicture}[->,>=stealth',shorten >=1pt,auto,node
    distance=2.0cm,initial text=,scale=0.9,transform shape]
    \tikzstyle{every node}=[font=\small] \tikzstyle{every
      state}=[fill=white,shape=circle,inner sep=.5mm,minimum size=6mm]
    \path[use as bounding box] (-1,-.5) rectangle (5.8,1.5); 
    \node[state,initial] (s0) at (0,0) {};
    \node[state] (s1) at (3,0) {};
    \path (s0) edge [solid] node [above,sloped] {receive, $2$} (s1);
    \path (s1) edge [solid,bend right=45] node [above,sloped] {deliver, $3$} (s0);
    \path (s1) edge [solid,loop right] node [right] {check, $1$} (s1);
  \end{tikzpicture}}
\subfigure[]{
    \label{fig:introimpl:2}
  \begin{tikzpicture}[->,>=stealth',shorten >=1pt,auto,node
    distance=2.0cm,initial text=,scale=0.9,transform shape]
    \tikzstyle{every node}=[font=\small] \tikzstyle{every
      state}=[fill=white,shape=circle,inner sep=.5mm,minimum size=6mm]
    \path[use as bounding box] (-1,-0.5) rectangle (4,1.5); 
    \node[state,initial] (s0) at (0,0) {};
    \node[state] (s1) at (3,0) {};
    \path (s0) edge [solid] node [above,sloped] {receive, $4$} (s1);
    \path (s1) edge [solid,bend right=45] node [above,sloped] {deliver, $3$} (s0);
  \end{tikzpicture}}
\subfigure[]{
    \label{fig:introimpl:3}
  \begin{tikzpicture}[->,>=stealth',shorten >=1pt,auto,node
    distance=2.0cm,initial text=,scale=0.9,transform shape]
    \tikzstyle{every node}=[font=\small] \tikzstyle{every
      state}=[fill=white,shape=circle,inner sep=.5mm,minimum size=6mm]
    \path[use as bounding box] (-1,-1.2) rectangle (5.8,1.5); 
    \node[state,initial] (s0) at (0,0) {};
    \node[state] (s1) at (3,0) {};
    \node[state] (s2) at (5,0) {};
    \path (s0) edge [solid] node [above,sloped] {receive, $3$} (s1);
    \path (s1) edge [solid,bend right=45] node [above,sloped] {deliver, $3$} (s0);
    \path (s1) edge [solid] node [above] {check, $1$} (s2);
    \path (s2) edge [solid,bend left=25] node [below,sloped] {deliver, $3$} (s0);
  \end{tikzpicture}}
  \subfigure[]{
    \label{fig:introimpl:4}
  \begin{tikzpicture}[->,>=stealth',shorten >=1pt,auto,node
    distance=2.0cm,initial text=,scale=0.9,transform shape]
    \tikzstyle{every node}=[font=\small] \tikzstyle{every
      state}=[fill=white,shape=circle,inner sep=.5mm,minimum size=6mm]
    \path[use as bounding box] (-1,-1.2) rectangle (4,1.5); 
    \node[state,initial] (s0) at (0,0) {};
    \node[state] (s1) at (3,0) {};
    \path (s0) edge [solid] node [above,sloped] {receive, $2$} (s1);
    \path (s1) edge [solid,bend right=45] node [above,sloped] {deliver, $3$} (s0);
  \end{tikzpicture}}
\caption{Four implementations of the simple email system in
  Figure~\ref{fig:mtsquantities}.}
  \label{fig:introimpl}
\end{figure}

\medskip%
In the above examples, the transition weights have expressed the time
used to perform the associated action.  However our formalism is
abstract enough to also model other quantitative aspects such as
\eg~energy consumption or financial aspects.  For instance,
Figure~\ref{fig:wiper} presents a simple electronic wiper control
component for a car, with a normal mode and an optional fast
mode. Integer intervals express the allowed energy consumption of each
action (using abstract energy units).

Depending on the precise application of our quantitative formalism,
there are a few choices which one has to make.  One such choice is the
precise definition of quantitative refinement, as the way quantitative
discrepancies between specifications is measured \eg~depends on whether
differences accumulate over time or the interest more lies in the
maximal individual differences.  Another choice is how to combine
quantities during structural composition: when modeling \eg~energy
consumption, they should be added; when modeling timing constraints,
some form of conjunction should be used.  To simplify presentation, we
develop the theory in this paper for one specific kind of quantitative
refinement and one specific choice of composition; a more general
treatment is deferred to future work.

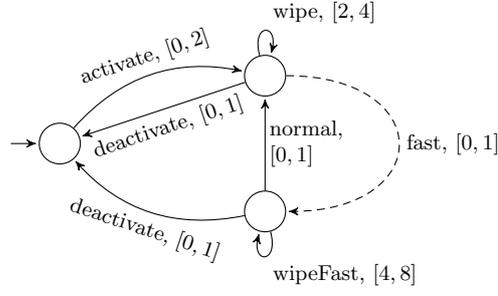
\begin{figure}[tp]
  \centering
  \begin{tikzpicture}[->,>=stealth',shorten >=1pt,auto,node
    distance=2.0cm,initial text=,scale=0.9,transform shape]
    \tikzstyle{every node}=[font=\small] \tikzstyle{every
      state}=[fill=white,shape=circle,inner sep=.5mm,minimum size=6mm]
    \node[state,initial] (s0) at (0,0) {};
    \node[state] (s1) at (3,1) {};
    \node[state] (s2) at (3,-1) {};
    \path (s0) edge [solid,bend left] node [above,sloped] {activate,
      $[0,2]$} (s1); 
    \path (s1) edge [solid] node [below,sloped] {deactivate, $[0,1]$} (s0);
    \path (s1) edge [solid,loop above] node [above right] {wipe, $[2,4]$} (s1);
    \draw[densely dashed,rounded corners] (s1) .. node [right] {fast,
      $[0,1]$} controls (5.5,1) and (5.5,-1) .. (s2); 
    \path (s2) edge [solid] node [right,xshift=-1mm] {
      \begin{tabular}{l}
normal,\\ $[0,1]$
\end{tabular}} (s1);
    \path (s2) edge [solid,loop below] node [below right] {wipeFast,
      $[4,8]$} (s2); 
    \path (s2) edge [solid,bend left] node [below,sloped] {deactivate,
      $[0,1]$} (s0); 
  \end{tikzpicture} 
  \caption{Weighted modal transition system modeling a simple wiper
    control component of a car.} 
  \label{fig:wiper}
\end{figure}

\medskip
To facilitate quantitative reasoning on specifications and
implementations, we introduce a real-valued \emph{distance} between
specifications such that perfect refinement corresponds to distance $0$,
small quantitative discrepancies give rise to small distances, and
differences in the discrete control structure correspond to distance
$\infty$.  For the examples in Figs.~\ref{fig:mtsquantities}
and~\ref{fig:introimpl}, we will hence deduce the following chain of
decreasing distances:
\begin{equation*}
  \infty = d(I_1,S) > d(I_2,S) > d(I_3,S) > d(I_4,S) = 0
\end{equation*}
Our distance is \emph{discounted} in the sense that behaviors which
occur $d$ steps in the future are discounted by a factor $\lambda^d$,
where $\lambda$ with $0< \lambda< 1$ is a fixed discounting factor.

Using a reduction to discounted games~\cite{DBLP:conf/cocoon/ZwickP95},
we show that this so-called \emph{modal} distance is computable in
\NP~$\cap$~\coNP.  As any specification can be seen as the (generally
infinite) set of implementations which are in perfect refinement, we
also have a natural notion of so-called \emph{thorough} distance between
specifications which is given by the (Hausdorff) distance between their
implementation sets; we show that computing through distances is
\EXPTIME-hard.

Replacing Boolean refinement by distances
has an impact on operations between specifications.  As a second
contribution of this paper, we propose quantitative versions of
structural composition and quotient which inherit the good properties
from the Boolean setting.  We also propose a new notion of
\emph{relaxation} which is inherent to the quantitative framework and
allows \eg~to calibrate the quotient operator: If the overall
specification is too restrictive with respect to a partial
implementation to synthesize a meaningful specification of the missing
components, the overall specification may be relaxed to facilitate a
better quotient.

However, there is no free lunch, and working with distances has a price:
some of the properties of logical conjunction and determinization are
not preserved in our quantitative setting. More precisely, conjunction
is not the greatest lower bound with respect to refinement distance as
it is in the Boolean setting, and deterministic overapproximation is too
coarse.  In fact we show that this is a fundamental limitation of
\emph{any} reasonable quantitative specification formalism.

Our final contribution consists of showing that a quantitative
interpretation of Hennessy-Milner logic provides a logical
characterization which is sound with respect to refinement distance and
complete for the disjunction-free fragment.

\paragraph{Related work.}

The objective of the paper is to propose a new complete quantitative
modal specification theory, which exploits a notion of distance between
specifications.  This distance builds on previous work of some of the
authors~\cite{journals/cai/FahrenbergLT10,DBLP:conf/fsttcs/FahrenbergLT11,conf/qapl/FahrenbergTL11,journals/tcs/LarsenFT11,Thrane11Thesis,journals/jlp/ThraneFL10}. For
the sake of completeness, we briefly put it in perspective with other
notions of distances proposed, particularly but not exclusively for
probabilistic systems, in recent years.  These
include~\cite{thesis/Breugel94,Breugel95-Metric} which develop a theory
of \emph{metric transition systems} and introduce the notion of compact
branching,
\cite{DBLP:conf/concur/Alfaro03,DBLP:conf/icalp/AlfaroHM03,DBLP:journals/tcs/DesharnaisGJP04,thesis/Majumdar03}
which introduce discounting distances for Markov decision processes,
and~\cite{DBLP:journals/corr/abs-0809-4326,DBLP:journals/lmcs/AlfaroMRS08}
which generalize these to a game setting.

For a non-probabilistic setting of metric transition systems (different
from van Breugel's), notions of discounting linear and branching
distances are developed in~\cite{DBLP:journals/tse/AlfaroFS09}, and an
important theoretical contribution
is~\cite{DBLP:journals/tcs/BonsangueBR98} which develops a theory of
directed distances, or \emph{hemimetrics} as they have come to be
called, and relate completion of hemimetric spaces to Yoneda embeddings
(see also~\cite{journals/rsmfm/Lawvere73,journals/rcm/Lawvere86}).
Another, language-based approach to quantitative verification, related
to the theory of semiring-weighted
automata~\cite{DBLP:journals/tcs/DrosteG07,book/DrosteKV09,DBLP:journals/tcs/DrosteR09},
can be found
in~\cite{DBLP:journals/tcs/CernyHR12,DBLP:journals/corr/abs-1007-4018}.

\paragraph{Structure of the paper.} 

The paper starts by introducing our quantitative formalism which has
weighted transition systems as implementations and weighted modal
transition systems as specifications.  In Section~\ref{se:distances} we
introduce the distances we use for quantitative comparison of both
implementations and specification, and Section~\ref{se:complex} provides
complexity results for the computation of these distances.
Section~\ref{se:relax} is devoted to a formalization of the notion of
relaxation which is of great use in quantitative design.  In
Section~\ref{se:hull} we see some inherent limitations of the
quantitative approach, and Section~\ref{se:parcomp} shows that
structural composition works as expected in the quantitative framework
and links relaxation to quotients.  Section~\ref{se:logics} finishes the
paper by providing logical characterizations of refinement distance.

\section{Weighted Modal Transition Systems}
\label{se:wmts}

In this section we present the formalism we use for implementations and
specifications.  As implementations we choose the model of
\emph{weighted transition systems}, \ie~labeled transition systems with
integer weights at transitions.  Specifications both have a \emph{modal}
dimension, specifying discrete behavior which \emph{must} be implemented
and behavior which \emph{may} be present in implementations, and a
\emph{quantitative} dimension, specifying intervals of weights on each
transition within are permissible for an implementation.

Let $\I=\big\{[ x, y]\bigmid x\in \Int\cup\{ -\infty\}, y\in \Int\cup\{
\infty\}, x\le y\big\}$ be the set of closed extended-integer intervals
and let $\Sigma$ be a finite set of actions.  Our set of
\emph{specification labels} is $\K= \Sigma\times \I$, pairs of actions
and intervals.  The set of \emph{implementation labels} is defined as
$\Imp \K= \Sigma\times\big\{[ x, x]\bigmid x\in \Int\big\}\approx
\Sigma\times \Int$. Hence a specification imposes labels and integer
intervals which constrain the possible weights of an implementation.

We define a partial order on $\I$ (representing inclusion of intervals)
by $[ x, y] \sqsubseteq [ x', y']$ if $x'\le x$ and $y\le y'$, and we
extend this order to specification labels by $( a, I)\sqsubseteq( a',
I')$ if $a= a'$ and $I\sqsubseteq I'$.  The partial order on $\K$ is
hence a \emph{refinement} order; if $k_1\sqsubseteq k_2$ for $k_1,
k_2\in \K$, then no more implementation labels are contained in $k_1$
than in $k_2$.

Specifications and implementations are defined as follows:

\begin{definition}
  A \emph{weighted modal transition system} (WMTS) is a quadruple $(
  S, s^0,\linebreak \mmayto, \mmustto)$ consisting of a set of states
  $S$ with an initial state $s^0\in S$ and \must ($\mmustto$) and \may
  ($\mmayto$) transition relations $\mmustto, \mmayto\subseteq S\times
  \K\times S$ such that for every $(s,k,s') \in \mmustto$ there is
  $(s,\ell,s') \in \mmayto$ where $k \sqsubseteq \ell$.  A WMTS is an
  \emph{implementation} if $\mmustto= \mmayto\subseteq S\times \Imp
  \K\times S$.
\end{definition}

Note the natural requirement that any required (\must) behavior is also
allowed (\may) above, and that implementations correspond to standard
integer-weighted transition systems, where all optional behavior and
positioning in the intervals has been decided on.

A WMTS is \emph{finite} if $S$ and $\mmayto$ (and hence also $\mmustto$)
are finite sets, and it is \emph{deterministic} if it holds that for any
$s\in S$ and $a\in \Sigma$, $\big( s,( a, I_1), t_1\big),\big( s,( a,
I_2), t_2\big)\in \mmayto$ imply $I_1= I_2$ and $t_1= t_2$.  Hence a
deterministic specification allows at most one transition under each
discrete action from every state.  In the rest of the paper we will
write $s\mayto[ k] s'$ for $( s, k, s')\in \mmayto$ and similarly for
$\mmustto$, and we will always write $S=( S, s^0, \mmayto, \mmustto)$ or
$S_i=( S_i, s^0_i, \mmayto_i, \mmustto_i)$ for WMTS and $I=( I, i^0,
\mmustto)$ for implementations.  Note that an implementation is just a
usual integer-weighted transition system.

Our theory will work with infinite WMTS, though we will require them to
be \emph{compactly branching}.  This is a natural generalization of the
standard requirement on systems to be \emph{finitely branching} which
was first used in~\cite{Breugel95-Metric}; see Definition~\ref{de:comp}
below.

The implementation semantics of a specification is given through modal
refinement, as follows:

\begin{definition}
  A \emph{modal refinement} of WMTS $S_1$, $S_2$
  is a relation $R\subseteq S_1\times S_2$ such that for any $(s_1,
  s_2)\in R$
  \begin{itemize}
  \item whenever $s_1\mayto[ k_1]_1 t_1$ for some $k_1\in \K$, $t_1\in
    S_1$, then there exists $s_2\mayto[ k_2]_2 t_2$ for some $k_2\in
    \K$, $t_2\in S_2$, such that $k_1 \sqsubseteq k_2$ and $( t_1,
    t_2)\in R$,
  \item whenever $s_2\mustto[ k_2]_2 t_2$ for some $k_2\in \K$,
    $t_2\in S_2$, then there exists $s_1\mustto[ k_1]_1 t_1$ for some
    $k_1\in \K$, $t_1\in S_1$, such that $k_1\sqsubseteq k_2$ and $(
    t_1, t_2)\in R$.
  \end{itemize}
  We write $S_1\le_m S_2$ if there is a modal refinement relation $R$
  for which $( s^0_1, s^0_2)\in R$.
\end{definition}

Hence in such a modal refinement, behavior which is required in $S_2$ is
also required in $S_1$, no more behavior is allowed in $S_1$ than in
$S_2$, and the quantitative requirements in $S_1$ are refinements of the
ones in $S_2$.
The implementation semantics of a specification can then be defined as
the set of all implementations which are also refinements:

\begin{definition}
  The \emph{implementation semantics} of a WMTS $S$ is the set
  $\llbracket S\rrbracket=\{ I \mid I \le_m S ~\text{and}~ I ~\text{is
    an implementation}\}$.
\end{definition}

This is conform with the intuition developed in the introduction: if
$I\in\llbracket S\rrbracket$, then any (reachable) behavior $i\mustto[
a, x] j$ in $I$ must be allowed by a matching transition $s\mayto[{
  a,[ l, r]}] t$ in $S$ with $l\le x\le r$; correspondingly, any
(reachable) required behavior $s\mustto[{ a,[ l, r]}] t$ in $S$ must be
implemented by a matching transition $i\mustto[ a, x] j$ in $I$ with
$l\le x\le r$.

\section{Thorough and Modal Refinement Distances}
\label{se:distances}

For the quantitative specification formalism we have introduced in the
last section, the standard Boolean notions of satisfaction and
refinement are too fragile.  To be able to reason not only whether a
given quantitative implementation satisfies a given quantitative
specification, but also to what extent, we introduce a notion of
\emph{distance} between both implementations and specifications.

We recall some terminology. Let $\Realnn\cup\{ \infty\}$ denote the
extended positive reals, let $X$ be a set and $d: X\times X\to
\Realnn\cup\{ \infty\}$.  Then $d$ is called
\begin{itemize}
\item a \emph{hemimetric} if $d( x, x)= 0$ for all $x\in X$
  (indiscernibility of identicals) and $d( x, y)+ d( y, z)\ge d( x,
  z)$ for all $x, y, z\in X$ (triangle inequality);
\item a \emph{pseudometric} if it is a hemimetric and additionally,
  $d( x, y)= d( y, x)$ for all $x, y\in X$ (symmetry);
\item a \emph{metric} if it is a pseudometric and additionally, $d( x,
  y)= 0$ implies $x= y$ for all $x, y\in X$ (identity of
  indiscernibles)
\end{itemize}
Note that as our (hemi-, pseudo-)metrics may take the values $\infty$,
some authors will refer to them as \emph{extended} (hemi-,
pseudo-)metrics.

The \emph{symmetrization}\label{pg:distsym} of a
hemimetric $d$ is the pseudometric $\bar d: X\times X\to \Realnn\cup\{
\infty\}$ given by $\bar d( x, y)= \max( d( x, y), d( y, x))$; this is
the smallest of all pseudometrics $d'$ on $X$ for which $d\le d'$.
Given hemimetrics $d$ on $X$ and $d'$ on another set $X'$, the
\emph{product distance} $D$ on $X\times X'$ is defined by $D(( x, x'),(
y, y'))= d( x, y)+ d( x', y')$.

\medskip
We first define the distance between \emph{implementations}; for this we
introduce a distance on implementation labels by
\begin{equation}
  \label{eq:d_imp}
  \idist\big(( a_1, x_1),( a_2, x_2)\big)=\left\{
  \begin{array}{cl}
    \infty &\quad\text{if } a_1\ne a_2, \\
    | x_1- x_2| &\quad\text{if } a_1= a_2.
  \end{array}
  \right.
\end{equation}
In the rest of the paper, let $\lambda\in \Real$ with $0 <
\lambda< 1$ be a \emph{discounting factor}.

\begin{definition}
  \label{de:acc.dist}
  The \emph{implementation distance} $\bdist: I_1\times I_2\to
  \Realnn\cup\{ \infty\}$ between the states of implementations $I_1$
  and $I_2$ is the least fixed point of the equations
  \begin{equation*}
    \bdist( i_1, i_2)= \max\left\{
      \begin{aligned}
        &\adjustlimits \sup_{ i_1\tto[ k_1]_1 j_1} \inf_{ i_2\tto[
          k_2]_2 j_2} \idist( k_1, k_2)+ \lambda \bdist( j_1, j_2), \\
        &\adjustlimits \sup_{ i_2\tto[ k_2]_2 j_2} \inf_{ i_1\tto[
          k_1]_1 j_1} \idist( k_1, k_2)+ \lambda \bdist( j_1, j_2).
      \end{aligned}
    \right.
  \end{equation*}
  We define $\bdist(I_1,I_2) = \bdist(i_1^0, i_2^0)$.
\end{definition}

\begin{lemma}
  \label{le:impdistmet}
  The implementation distance is well-defined, and is a pseudometric.
\end{lemma}

\begin{proof}
  Except for the symmetrizing max operation, this is precisely the
  \emph{accumulating branching distance}
  from~\cite{journals/tcs/LarsenFT11,journals/jlp/ThraneFL10}.  Because
  of $\lambda< 1$, the equations above define a contraction (with
  Lipschitz constant $\lambda$), so the Banach fixed point theorem (for
  extended metric spaces) applies.  Hence besides the fixed point
  $\bdist( i_1, i_2)= \infty$, the contraction has at most one other
  fixed point, \ie~there exists indeed a unique least fixed point.  We
  refer to~\cite{journals/tcs/LarsenFT11} for a more detailed proof.

  Symmetry of $d$ is clear, and so is the property $d( i, i)= 0$.  The
  triangle inequality can be shown inductively,
  \cf~\cite{journals/tcs/LarsenFT11}. \qed
\end{proof}

We remark that besides this accumulating distance, other interesting
system distances may be defined depending on the application at hand,
\cf~\cite{journals/jlp/ThraneFL10,journals/cai/FahrenbergLT10,DBLP:conf/fsttcs/FahrenbergLT11},
but we concentrate here on this distance and leave a generalization to
other distances for future work.

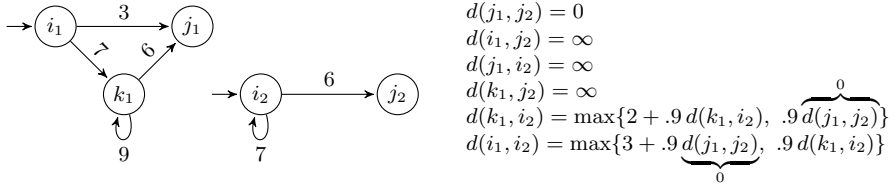
\begin{figure}[tpb]
  \begin{center}
    \begin{tikzpicture}[->,>=stealth',shorten >=1pt,auto,node
      distance=2.0cm,initial text=,scale=0.9,transform shape]
      \tikzstyle{every node}=[font=\small] \tikzstyle{every
        state}=[fill=white,shape=circle,inner sep=.5mm,minimum size=6mm]
      \path[use as bounding box] (0,-0.5) rectangle (12,1.3); 
      \node[state,initial] (s) at (0,1) {$i_1$};
      \node[state] (s1) at (2,1) {$j_1$};
      \node[state] (s2) at (1,0) {$k_1$};
      \path (s) edge [solid] node [above,sloped] {$3$} (s1);
      \path (s) edge [solid] node [above,sloped] {$7$} (s2);
      \path (s2) edge [solid] node [above,sloped] {$6$} (s1);
      \path (s2) edge [solid,loop below] node [below,sloped] {$9$} (s2);
      
      \node[state,initial] (t) at (3,0) {$i_2$};
      \node[state] (t1) at (5,0) {$j_2$};
      \path (t) edge [solid] node [above,sloped] {$6$} (t1);
      \path (t) edge [solid,loop below] node [below,sloped] {$7$} (t);
      
      \node (text) at (9.4,0) 
      {\begin{minipage}{6.8cm}
          $\bdist(j_1,j_2) = 0$ \\
          $\bdist(i_1,j_2) = \infty$ \\
          $\bdist(j_1,i_2) = \infty$  \\
          $\bdist(k_1,j_2) = \infty$  \\[-4mm]
          $\bdist(k_1,i_2) = \max\{2+.9\, \bdist(k_1,i_2),\ 
          .9\overbrace{\bdist(j_1,j_2)}^0 \}$ \\
          $\bdist(i_1,i_2) = \max\{3+.9\underbrace{\bdist(j_1,j_2)}_0,\ 
          .9\, \bdist(k_1,i_2) \}$ \\
        \end{minipage}
      };
    \end{tikzpicture}
  \end{center}
  \caption{Two weighted transition systems with branching distance
    $\bdist( I_1, I_2)= 18$.}
  \label{fig:distance}
\end{figure}

\begin{example}
  Consider the two implementations $I_1$ and $I_2$ in
  Figure~\ref{fig:distance} with a single action (elided for simplicity)
  and with discounting factor $\lambda = .9$.  The equations in the
  illustration have already been simplified by removing all expressions
  that evaluate to $\infty$.  What remains to be done is to compute the
  least fixed point of the equation $\bdist(k_1,i_2) = \max\big\{ 2 +
  .9\, \bdist(k_1,i_2), 0\big\}$.  Clearly $0$ is not a fixed point, and
  solving the equation $\bdist( k_1, i_2)= 2+ .9\,\bdist( k_1, i_2)$
  gives $\bdist(k_1,i_2) = 20$. Hence $\bdist(i_1,i_2) = \max\{3,
  .9\cdot 20\} = 18$.
\end{example}

Note that the interpretation of the distance between two implementations
depends entirely on the application one has in mind; but it can easily
be shown~\cite{journals/jlp/ThraneFL10} that the distance between two
implementations is zero if and only if they are \emph{weighted
  bisimilar}.  The intuition is then that the smaller the distance, the
closer the implementations are to being bisimilar.

To lift the implementation distance to specifications, we need first to
consider the distance between \emph{sets} of implementations.  Given
implementation sets $\mathcal I_1$, $\mathcal I_2$, we define
\begin{equation*}
  \bdist( \mathcal I_1, \mathcal I_2)= \adjustlimits \sup_{ I_1\in
    \mathcal I_1} \inf_{ I_2\in \mathcal I_2} \bdist( I_1, I_2)
\end{equation*}
Note that in case $\mathcal I_2$ is finite, we have that for all
$\epsilon\ge 0$, $\bdist( \mathcal I_1, \mathcal I_2)\le \epsilon$ if
and only if for each implementation $I_1\in \mathcal I_1$ there exists
$I_2\in \mathcal I_2$ for which $\bdist( I_1, I_2)\le \epsilon$, hence
this is quite a natural notion of distance.  Especially, $\bdist(
\mathcal I_1, \mathcal I_2)= 0$ if $\mathcal I_1$ is a subset of
$\mathcal I_2$ up to bisimilarity.  For infinite $\mathcal I_2$, we have
the slightly more complicated property that $\bdist( \mathcal I_1,
\mathcal I_2)\le \epsilon$ if and only if for all $\delta> 0$ and any
$I_1\in \mathcal I_1$, there is $I_2\in \mathcal I_2$ for which $\bdist(
I_1, I_2)\le \epsilon+ \delta$.

Also remark the similarity of this definition to the one of
\emph{Hausdorff distance} between subsets of a metric space, see
\eg~\cite[Sect.~3.16]{book/AliprantisB07}.  Crucially however, our
distance is missing the symmetrizing max operation of Hausdorff
distance, hence it is \emph{asymmetric}.  We may well have $\bdist(
\mathcal I_1, \mathcal I_2)\ne \bdist( \mathcal I_2, \mathcal I_1)$ and
will thus prefer to speak of the distance \emph{from $\mathcal I_1$ to
  $\mathcal I_2$} rather than \emph{between} $\mathcal I_1$ and
$\mathcal I_2$.  We lift this distance to specifications as follows:

\begin{definition}
  The \emph{thorough refinement distance} between WMTS $S_1$ and $S_2$
  is defined as $\tdist( S_1, S_2)= \bdist\big(\llbracket S_1\rrbracket,
  \llbracket S_2\rrbracket\big)$.  We write $S_1\le_t^\epsilon S_2$ if
  $\tdist( S_1, S_2)\le \epsilon$.
\end{definition}

\begin{lemma}
  The thorough refinement distance is a hemimetric.
\end{lemma}

\begin{proof}
  To show that $d_t( S, S)= 0$ is trivial, and the triangle inequality
  $d_t( S_1, S_2)+ d_T( S_2, S_3)\ge d_t( S_1, S_3)$ follows like in the
  proof of~\cite[Lemma~3.72]{book/AliprantisB07}. \qed
\end{proof}

Indeed this permits us to measure incompatibility of specifications;
intuitively, if two specifications have thorough distance $\epsilon$,
then any implementation of the first specification can be matched by an
implementation of the second up to $\epsilon$.  Also observe the special
case where $S_1= I_1$ is an implementation: then $\tdist( I_1, S_2)=
\inf_{ I_2\in \llbracket S_2\rrbracket} \bdist( I_1, I_2)$, which
measures how close $I_1$ is to satisfy the specification $S_2$.

\medskip
To facilitate computation and comparison of refinement distance, we
introduce \emph{modal} refinement distance as an overapproximation.  We
will show in Theorem~\ref{pr:exptime-hard} below that similarly to the
Boolean setting~\cite{DBLP:conf/ictac/BenesKLS09}, computation of
thorough refinement distance is \EXPTIME-hard, whereas modal refinement
distance is computable in \NP~$\cap$~\coNP.

First we generalize the distance on implementation labels from
Equation~\eqref{eq:d_imp} to specification labels, again using a
Hausdorff-type construction.  For $k, \ell \in \K$ we define
\begin{equation*}
  \kdist( k, \ell)= \adjustlimits \sup_{ k'\sqsubseteq k, k'\in \Imp{}\,}
  \inf_{\, \ell'\sqsubseteq \ell, \ell'\in \Imp{}} \idist( k', \ell').
\end{equation*}
Note that $\kdist$ is asymmetric, and that $\kdist( k, \ell)= 0$ if and
only if $k\sqsubseteq \ell$.  Also, $\kdist( k, \ell)= \idist( k, \ell)$
for all $k, \ell\in \Imp \K$.  In more elementary terms, we can express
$\kdist$ as follows:
\begin{align*}
  \kdist\big(( a_1, I_1),( a_2, I_2)\big) &= \infty \quad \text{if }
  a_1\ne a_2 \\
  \kdist\big(( a,[ x_1, y_1]),( a,[ x_2, y_2])\big) &= \max(
  x_2 - x_1, y_1 - y_2, 0)
\end{align*}

\begin{definition}
  \label{de:acc.mo.dist}
  Let $S_1$, $S_2$ be WMTS.  The \emph{modal refinement distance}
  $\mdist: S_1\times S_2\to \Realnn\cup\{ \infty\}$ from states of
  $S_1$ to states of $S_2$ is the least fixed point of the equations
  \begin{equation*}
    \mdist( s_1, s_2)= \max\left\{
      \begin{aligned}
        &\adjustlimits \sup_{ s_1\mayto[ k_1]_1 t_1} \inf_{ s_2\mayto[
          k_2]_2 t_2} \kdist( k_1, k_2)+ \lambda \mdist( t_1, t_2), \\
        &\adjustlimits \sup_{ s_2\mustto[ k_2]_2 t_2} \inf_{
          s_1\mustto[ k_1]_1 t_1} \kdist( k_1, k_2)+ \lambda \mdist(
        t_1, t_2).
      \end{aligned}
    \right.
  \end{equation*}
  We define $\mdist(S_1,S_2) = \mdist(s_1^0, s_2^0)$, and we write
  $S_1\le_m^\epsilon S_2$ if $\mdist( S_1, S_2)\le \epsilon$.
\end{definition}

\begin{lemma}
  The modal refinement distance is well-defined, and is a hemimetric.
\end{lemma}

\begin{proof}
  Like in the proof of Lemma~\ref{le:impdistmet}, the argument for
  existence of a unique least fixed point to the defining equations is
  that they define a contraction.  The triangle inequality can again be
  shown inductively, and the property $\mdist( s, s)= 0$ is clear. \qed
\end{proof}

We can now give a precise definition of compact branching; for this we
need the notions of symmetrization of a hemimetric and of product
distance as defined on page~\pageref{pg:distsym}.

\begin{definition}
  \label{de:comp}
  A WMTS $S$ is said to be \emph{compactly branching} if the sets $\{(
  s', k)\mid s\mayto[ k] s'\}, \{( s', k)\mid s\mustto[ k] s'\}\subseteq
  S\times \K$ are compact under the symmetrized product distance $\bar
  d_m\times \bar d_\K$ for every $s\in S$.
\end{definition}

The notion of compact branching was first introduced, for a formalism of
\emph{metric transition systems}, in~\cite{Breugel95-Metric}.  It is a
natural generalization of the standard requirement on transition systems
to be \emph{finitely branching} to a distance setting; we will need it
for the property that continuous functions defined on the sets $\{( s',
k)\mid s\mayto[ k] s'\}, \{( s', k)\mid s\mustto[ k] s'\}\subseteq
S\times \K$, for some $s\in S$, attain their infimum and supremum, see
Lemma~\ref{le:family} and its proof below.

Thus, we shall henceforth assume all our WMTS to be compactly
branching.  The following lemma sets up some sufficient conditions for
this to be the case.

\begin{lemma}
  Let $S$ be a WMTS and define the sets $L_i( s, a)$, $U_i( s, a)$ for
  all $s\in S$, $a\in \Sigma$ and $i\in\{ 1, 2\}$ by
  \begin{align*}
    L_1( s, a) &= \{ l\mid s\mayto[{ a,[ l, r]}] s'\}, &
    L_2( s, a) &= \{ l\mid s\mustto[{ a,[ l, r]}] s'\}, \\
    U_1( s, a) &= \{ r\mid s\mayto[{ a,[ l, r]}] s'\}, 
    & U_2( s, a) &= \{ r\mid s\mustto[{ a,[ l, r]}] s'\}.
  \end{align*}
  Then $S$ is compactly branching if
  \begin{itemize}
  \item for all $s\in S$, any Cauchy sequence $( s'_n)_{ n\in \Nat}$ in
    $\{ s'\mid s\mayto s'\}$ (with pseudometric $\bar d_m$) has $\lim_{
      n\to \infty} s_n\in\{ s'\mid s\mayto s'\}$, and likewise, any
    Cauchy sequence $( s'_n)_{ n\in \Nat}$ in $\{ s'\mid s\mustto s'\}$
    has $\lim_{ n\to \infty} s_n\in\{ s'\mid s\mustto s'\}$, and
  \item for all $s\in S$, $a\in \Sigma$ and $i\in\{ 1, 2\}$, $L_i$ is
    finite or $-\infty\in L_i$, and $U_i$ is finite or $\infty\in U_i$.
  \end{itemize}
\end{lemma}

Note that the first property mimicks (and generalizes) standard
properties of finite branching and \emph{saturation},
\cf~\cite[Sect.~3.3]{DBLP:journals/toplas/Sangiorgi09}.  The intuition
is that if $s$ has (either \may or \must) transitions to a converging
sequence of states, then it also has a transition to the limit.

\begin{proof}
  The first condition implies that the sets $\{ s'\in S\mid s\mayto
  s'\}$ and $\{ s'\in S\mid s\mustto s'\}$ are compact in the
  pseudometric $\bar d_m$ for all $s\in S$.  By Tychonoff's theorem,
  products of compact sets are compact, so we need only show that the
  second condition implies that the sets $\{ k\in \K\mid s\mayto[ k]
  s'\}$ and $\{ k\in \K\mid s\mustto[ k] s'\}$ are compact in the
  pseudometric $\bar d_\K$ for every $s\in S$.

  Let $s\in S$.  By definition of $\kdist$, the sets $\{ k\mid s\mayto[
  k] s'\}$, $\{ k\mid s\mustto[ k] s'\}$ fall into connected components
  $\{ I\mid s\mayto[ a, I] s'\}$, $\{ I\mid s\mustto[ a, I] s'\}$ for
  all $a\in \Sigma$, hence the former are compact if and only if all the
  latter are.  These in turn are compact if and only if the four sets
  $L_i$, $U_i$ in the lemma, collecting lower and upper bounds of
  intervals, are compact.  Now interval bounds are extended integers, so
  a sequence in $L_i$ or $U_i$ converges if and only if it is eventually
  stable or goes towards $-\infty$ or $\infty$.  If the sets are finite,
  eventual stability is the only option; if they are infinite, they need
  to include the limit points $-\infty$ (for the lower interval bounds
  in $L_i$) or $\infty$ (for the upper interval bounds in $U_i$). \qed
\end{proof}

There is a powerful proof technique introduced for branching
distances between implementations in~\cite{journals/jlp/ThraneFL10}
that we here extend to modal refinement distance.  We define a
\emph{modal refinement family} as an $\Realnn$-indexed family of
relations $R=\{ R_\epsilon\subseteq S_1\times S_2\mid \epsilon\ge 0\}$
such that for any $\epsilon$ and any $(s_1, s_2)\in R_\epsilon$,
\begin{itemize}
\item whenever $s_1\mayto[ k_1] t_1$ for some $k_1 \in \K$, $t_1 \in
  S_1$, then there exists $s_2\mayto[ k_2] t_2$ for some $k_2 \in \K$,
  $t_2 \in S_2$, such that $\kdist( k_1, k_2)\le \epsilon$ and $( t_1,
  t_2)\in R_{ \epsilon'}$ for some $\epsilon'\le \lambda^{ -1}\big(
  \epsilon- \kdist( k_1, k_2)\big)$,
\item whenever $s_2\mustto[ k_2] t_2$ for some $k_2 \in \K$, $t_2 \in
  S_2$, then there exists $s_1\mustto[ k_1] t_1$ for some $k_1 \in
  \K$, $t_1 \in S_1$, such that $\kdist( k_1, k_2)\le \epsilon$ and $(
  t_1, t_2)\in R_{ \epsilon'}$ for some $\epsilon'\le \lambda^{
    -1}\big( \epsilon- \kdist( k_1, k_2)\big)$.
\end{itemize}
Note that modal refinement families are
\begin{itemize}
\item \emph{upward closed} in the sense that $( s_1, s_2)\in R_\epsilon$
  implies that $( s_1, s_2)\in R_{ \epsilon'}$ for all $\epsilon'\ge
  \epsilon$, and
\item \emph{downward closed} in the sense that for any set $E\subseteq
  \Realnn$, if $(s_1,s_2) \in R_\epsilon$ for all $\epsilon\in E$, then
  also $(s_1,s_2) \in R_{\inf E}$.  This property follows from the
  assumption that our WMTS are compactly branching.
\end{itemize}

Following the proof strategy developed
in~\cite{journals/jlp/ThraneFL10} for implementations, we can show the
following characterization of modal refinement distance by modal
refinement families:

\begin{lemma}
  \label{le:family}
  $S_1\le_m^\epsilon S_2$ if and only if there is a modal refinement
  family $R$ with $( s_1^0, s_2^0)\in R_\epsilon \in R$.
\end{lemma}

\begin{proof}
  First, assume that $S_1\le_m^\epsilon S_2$, \ie~$d_m(s^0_1,s^0_2) \le
  \epsilon$, and define a relation family $R = \{ R_\delta \mid \delta
  \ge 0 \}$ by $R_\delta = \{ (s_1,s_2) \in S_1 \times S_2 \mid
  d_m(s_1,s_2) \le \delta \}$ for all $\delta \ge 0$, then
  $(s^0_1,s^0_2) \in R_\epsilon$ holds by assumption.  We show that $R$
  is a modal refinement family.  Let $(s_1,s_2) \in R_\delta$ for some
  $\delta \ge 0$, then by definition we know that $d_m(s_1,s_2) \le
  \delta$. Assume $s_1 \mayto[k_1]_1 t_1$. From $d_m(s_1,s_2) \le
  \delta$ we can infer that
  \begin{equation*}
    \inf_{ s_2\mayto[ k_2]_2 t_2} \kdist( k_1, k_2)+ \lambda
    \mdist( t_1, t_2) \le \delta.
  \end{equation*}
  Hence, because $S_2$ is compactly branching, there exists a
  may-transition $s_2 \mayto[k_2] t_2$ such that $\kdist( k_1, k_2) \le
  \delta$ and $d_m(t_1,t_2) \le \lambda^{-1}(\delta -
  \kdist(k_1,k_2))$. The latter implies that $(t_1,t_2) \in R_{\delta'}$
  for some $\delta' \le \lambda^{-1}(\delta - \kdist(k_1,k_2))$ which
  was to be shown. The argument for the other assertion for
  must-transitions is symmetric. This proves that there is a modal
  refinement family $R$ such that $(s^0_1,s^0_2) \in R_\epsilon \in R$.

  For the reverse direction, assume that $( s_1^0, s_2^0)\in R_\epsilon
  \in R$ for some modal refinement family $R = \{R_\epsilon \mid
  \epsilon \ge 0 \}$. We prove that $(s_1,s_2) \in R_\delta$, for some
  $\delta \ge 0$, implies $d_m(s_1,s_2) \le \delta$. The claim $S_1
  \le_m^\epsilon S_2$ then follows from the assumption $(s^0_1,s^0_2)
  \in R_\epsilon$.

  To this end, observe that the space of functions $\Delta = [S_1 \times
  S_2 \to \Realnn \cup \{\infty\}]$ forms a complete lattice, when the
  partial order $\le_\Delta$ is defined such that for $f, f' \in \Delta$, $f
  \le_\Delta f'$ iff $f(s_1,s_2) \le f'(s_1,s_2)$ for all $s_1\in S_1$, $s_2\in
  S_2$. Moreover, since $\max, \sup, \inf$ and $+$ are monotone, the
  function $D$ defined for all $f \in \Delta$ by
  \begin{equation*}
    D(f) = \max
      \begin{cases}
        &\adjustlimits \sup_{ s_1\mayto[ k_1]_1 t_1} \inf_{ s_2\mayto[
          k_2]_2 t_2} \kdist( k_1, k_2)+ \lambda f( t_1, t_2), \\
        &\adjustlimits \sup_{ s_2\mustto[ k_2]_2 t_2} \inf_{
          s_1\mustto[ k_1]_1 t_1} \kdist( k_1, k_2)+ \lambda f(
        t_1, t_2)
      \end{cases}
  \end{equation*}
  is a monotone endofunction on $\Delta$, hence by Tarski's fixed point
  theorem, $D$ has a least fixed point. Now let us define $h(s_1,s_2) =
  \inf\{\delta \mid (s_1,s_2) \in R_\delta \in R\}$, and since
  $R_\delta$ is downward closed, we have that $(s_1,s_2)\in
  R_{h(s_1,s_2)}$. By showing that $h$ is a pre-fixed point of $D$,
  \ie~that $D(h) \le_\Delta h$, we get that $(s_1,s_2)\in R_\delta$ implies
  that $d_m(s_1,s_2) \le \delta$, since $h(s_1,s_2) \le \delta$ and
  $d_m(s_1,s_2) \le h(s_1,s_2)$.

  Since $(s_1,s_2)\in R_{h(s_1,s_2)}$ every $s_1\mayto[k_1]s'_1$ can be
  matched by some $s_2\mayto[k_2] s'_2$ such that $\kdist(k_1,k_2) +
  \lambda \delta' \le h(s_1,s_2)$ for some $\delta'$ where $(s'_1,s'_2)
  \in R_{\delta'}$, implying $h(s'_1,s'_2) \le \delta'$, but then also
  $\kdist(k_1,k_2) + \lambda h(s'_1,s'_2) \le h(s_1,s_2)$. Similarly,
  every $s_2\mustto[k_2] s'_2$ has a match $s_1 \mustto[k_1] s'_1$ such
  that $\kdist(k_1,k_2) + \lambda h(s'_1,s'_2) \le h(s_1,s_2)$. Hence we
  have $D(h) \le_\Delta h$ which was to be shown. \qed
\end{proof}

The next theorems show that modal refinement distance indeed
overapproximates thorough refinement distance, and that it is exact for
deterministic WMTS.  Note that nothing general can be said about the
precision of the overapproximation in the nondeterministic case; as an
example observe the two specifications in
Figure~\ref{fig:incompleteness} for which $\tdist(S_1,S_2) = 0$ but
$\mdist(S_1,S_2) = \infty$.

\begin{figure}[tpb]
  \centering
  \subfigure
  {
    \raisebox{6mm}{
    \begin{tikzpicture}[->,>=stealth',shorten >=1pt,auto,node
      distance=2.0cm,initial text=,scale=0.8,transform shape]
      \tikzstyle{every node}=[font=\small] \tikzstyle{every
        state}=[fill=white,shape=circle,inner sep=.5mm,minimum size=6mm]
      \node[initial,state] (s1) at (0,0) {$s_1$};
      \node (S1) at (-1,0) {$S_1$};
      \node[state] (s2) at (2,0) {$t_1$};
      \path (s1) edge [densely dashed] node [above] {$a,[ 0, 1]$} (s2);
    \end{tikzpicture}}}
  \hspace{2mm}
  \subfigure
  {
    \begin{tikzpicture}[->,>=stealth',shorten >=1pt,auto,node
      distance=2.0cm,initial text=,scale=0.8,transform shape]
      \tikzstyle{every node}=[font=\small] \tikzstyle{every
        state}=[fill=white,shape=circle,inner sep=.5mm,minimum size=6mm]
      \node[initial,state] (s1) at (0,0) {$s_2$};
      \node (S2) at (-1,0) {$S_2$};
      \node[state] (s2) at (2,.8) {$t_2$};
      \node[state] (s3) at (2,-.8) {$t_3$};
      \path (s1) edge [densely dashed] node [above,sloped] {$a,[ 0, 0]$} (s2);
      \path (s1) edge [densely dashed] node [above,sloped] {$a,[ 1, 1]$} (s3);
    \end{tikzpicture}}
  
  \caption{\label{fig:incompleteness}%
    Incompleteness of modal refinement distance: $\tdist( S_1, S_2)= 0$,
    but $\mdist( S_1, S_2)= \infty$.}
\end{figure}
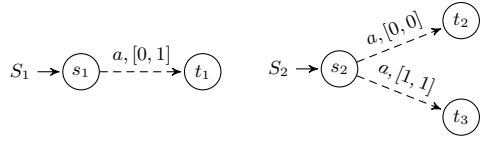

\begin{theorem}
  \label{th:dtledm}
  For WMTS $S_1$, $S_2$ we have $\tdist( S_1, S_2)\le \mdist( S_1,
  S_2)$.
\end{theorem}

\begin{proof}
  If $\mdist( S_1, S_2)=\infty$, we have nothing to prove.  Otherwise,
  let $R=\{ R_\epsilon\subseteq S_1\times S_2\mid \epsilon\ge 0\}$ be a
  modal refinement family which witnesses $\mdist( S_1, S_2)$, \ie~such
  that $( s_1^0, s_2^0)\in R_{ \mdist( S_1, S_2)}$, and let
  $I_1\in\llbracket S_1\rrbracket$.  We have to expose $I_2\in\llbracket
  S_2\rrbracket$ for which $\bdist( I_1, I_2)\le \mdist( S_1, S_2)$.

  Let $R_1\subseteq I_1\times S_1$ be a witness for $I_1\le_m S_1$,
  define $R'_\epsilon= R_1\circ R_\epsilon\subseteq I_1\times S_2$ for
  all $\epsilon\ge 0$, and let $R'=\{ R'_\epsilon\mid \epsilon\ge 0\}$.
  The states of $I_2=( I_2, i_2^0, \Imp \K, \mtto_{ I_2})$ are $I_2=
  S_2$ with $i_2^0= s_2^0$, and the transitions we define as
  follows:

  For any $i_1\tto[ k_1']_{ I_1} j_1$ and any $s_2\in S_2$ for which $(
  i_1, s_2)\in R'_\epsilon\in R'$ for some $\epsilon$, we have
  $s_2\mayto[ k_2]_2 t_2$ in $S_2$ with $\kdist( k_1', k_2)\le \epsilon$
  and $( j_1, t_2)\in R'_{ \epsilon'}\in R'$ for some $\epsilon'\le
  \lambda^{ -1}\big( \epsilon- \kdist( k_1', k_2)\big)$.  Write $k_1'=(
  a_1', x_1')$ and $k_2=\big( a_2,[ x_2, y_2]\big)$, then we must have
  $a_1'= a_2$.  Let
  \begin{equation}
    x_2'=
    \begin{cases}
      x_2 &\text{if } x_1'< x_2, \\
      x_1' &\text{if } x_2\le x_1'\le y_2, \\
      y_2 &\text{if } x_1'> y_2
    \end{cases}
    \label{eq:x_2'_app}
  \end{equation}
  and $k_2'=( a_2, x_2')$, and put $s_2\tto[ k_2']_{ I_2} t_2$ in $I_2$.
  Note that
  \begin{equation}
    \kdist( k_1', k_2')= \kdist( k_1', k_2).
    \label{eq:dk1'k2'=dk1'k2_app}
  \end{equation}

  Similarly, for any $s_2\mustto[ k_2]_2 t_2$ in $S_2$ and any $i_1\in I_1$
  with $( i_1, s_2)\in R'_\epsilon\in R'$ for some $\epsilon$, we have
  $i_1\tto[ k_1']_{ I_1} j_1$ with $\kdist( k_1', k_2)\le \epsilon$ and $(
  j_1, t_2)\in R'_{ \epsilon'}\in R'$ for some $\epsilon'\le \lambda^{
    -1}\big( \epsilon- \kdist( k_1', k_2)\big)$.  Write $k_1'=( a_1',
  x_1')$ and $k_2=( a_2,[ x_2, y_2])$, define $x_2'$ as
  in~\eqref{eq:x_2'_app} and $k_2'=( a_2, x_2')$, and put $s_2\tto[
  k_2']_{ I_2} t_2$ in $I_2$.

  We show that the identity relation $\id_{ S_2}=\{( s_2, s_2)\mid
  s_2\in S_2\}\subseteq S_2\times S_2$ witnesses $I_2\le_m S_2$.  Let
  first $s_2\tto[ k_2']_{ I_2} t_2$; we must have used one of the two
  constructions above for creating this transition.  In the first
  case, we have $s_2\mayto[ k_2]_2 t_2$ with $k_2'\sqsubseteq k_2$,
  and in the second case, we have $s_2\mustto[ k_2]_2 t_2$, hence also
  $s_2\mayto[ k_2]_2 t_2$, with the same property.  For a transition
  $s_2\mustto[ k_2]_2 t_2$ on the other hand, we have introduced
  $s_2\tto[ k_2']_{ I_2} t_2$ in the second construction above, with
  $k_2'\sqsubseteq k_2$.

  We also want to show that the family $R'$ is a witness for $\bdist(
  I_1, I_2)\le \mdist( S_1, S_2)$.  We have $( i_1^0, s_2^0)\in R'_{
    \mdist( S_1, S_2)}= R_1\circ R_{ \mdist( S_1, S_2)}$, so let $( i_1,
  s_2)\in R'_\epsilon\in R'$ for some $\epsilon\ge 0$.  For any
  $i_1\tto[ k_1']_{ I_1} j_1$ we have $s_2\mayto[ k_2]_2 t_2$ and
  $s_2\tto[ k_2']_{ I_2} t_2$ by the first part of our construction
  above, with $\kdist( k_1', k_2')= \kdist( k_1', k_2)\le \epsilon$ because
  of~\eqref{eq:dk1'k2'=dk1'k2_app}, and also $( j_1, t_2)\in R'_{
    \epsilon'}\in R'$ for some $\epsilon'\le \lambda^{ -1}\big(
  \epsilon- \kdist( k_1', k_2)\big)$.  For any $s_2\tto[ k_2']_{ I_2} t_2$,
  we must have used one of the constructions above to introduce this
  transition, and both give us $i_1\tto[ k_1']_{ I_1} j_1$ with $\kdist(
  k_1', k_2')\le \epsilon$ and $( j_1, t_2)\in R'_{ \epsilon'}\in R'$
  for some $\epsilon'\le \lambda^{ -1}\big( \epsilon- \kdist( k_1',
  k_2)\big)$. \qed
\end{proof}

The fact that modal refinement only equals thorough refinement for
deterministic specifications is well-known from the theory of modal
transition systems~\cite{Larsen89}, and the special case of $S_2$
deterministic is important, as it can be argued~\cite{Larsen89} that
indeed, deterministic specifications are sufficient for applications.

\begin{theorem} 
  \label{th:det-dteqdm}
  If $S_2$ is deterministic, then $\tdist( S_1, S_2)= \mdist( S_1,
  S_2)$.
\end{theorem}

\begin{proof}
  If $\tdist( S_1, S_2)= \infty$, we are done by
  Theorem~\ref{th:dtledm}.  Otherwise, let $R=\{ R_\epsilon\mid
  \epsilon\ge 0\}$ be the smallest relation family for which
  \begin{itemize}
  \item $( s_1^0, s_2^0)\in R_{ \tdist( S_1, S_2)}$ and
  \item whenever we have $( s_1, s_2)\in R_\epsilon\in R$, $s_1\mayto[
    a, I_1]_1 t_1$, and $s_2\mayto[ a, I_2]_2 t_2$,
    then $( t_1, t_2)\in R_{ \lambda^{ -1}( \epsilon- \kdist(( a, I_1),(
      a, I_2)))}$.
  \end{itemize}
  We show below that this definition makes sense (also that
  $\epsilon- \kdist\big(( a, I_1),( a, I_2)\big)\ge 0$ in all cases), and
  that $R$ is a modal refinement family.  We will use
  the convenient notation $( s_1, S_1)$ for the WMTS $S_1$ with
  initial state $s_1^0$ replaced by $s_1$, similarly for $( s_2,
  S_2)$.

  We first show inductively that for any pair of states $( s_1, s_2)\in
  R_\epsilon\in R$ we have $\tdist\big(( s_1, S_1),( s_2, S_2)\big)\le
  \epsilon$.  This is obviously the case for $s_1= s_1^0$ and $s_1=
  s_2^0$, so assume now that $( s_1, s_2)\in R_\epsilon\in R$ is such
  that $\tdist\big(( s_1, S_1),( s_2, S_2)\big)\le \epsilon$ and let
  $s_1\mayto[ a, I_1]_1 t_1$, $s_2\mayto[ a, I_2]_2 t_2$.  Let $( q_1',
  P_1')\in\llbracket( t_1, S_1)\rrbracket$ and $x_1\in I_1$.

  There is an implementation $( p_1, P_1)\in\llbracket( s_1,
  S_1)\rrbracket$ for which $p_1\tto[ a, x_1] q_1$ and such that $( q_1,
  P_1)\le_m( q_1', P_1')$.  Now
  \begin{equation*}
    \tdist\big(( p_1, P_1),( s_2, S_2)\big)\le \tdist\big(( p_1, P_1),(
    s_1, S_1)\big)+ \tdist\big(( s_1, S_1),( s_2, S_2)\big)\le \epsilon,
  \end{equation*}
  hence we must have $s_2\mayto[ a_2', I_2']_2 t_2'$ with $\kdist\big((
  a, x_1),( a_2', I_2')\big)\le \epsilon$.  But then $a_2'= a$, hence
  by determinism of $S_2$, $I_2= I_2'$ and $t_2= t_2'$.

  The above considerations hold for any $x_1\in I_1$, hence
  $\kdist\big(( a, I_1),( a, I_2)\big)\le \epsilon$.  Thus $\epsilon-
  \kdist\big(( a, I_1),( a, I_2)\big)\ge 0$, and the definition of $R$
  above is justified.  Now let $x_2\in I_2$ such that $\kdist\big(( a,
  x_1),( a, x_2)\big)= \kdist\big(( a, x_1),( a, I_2)\big)$, then there is
  an implementation $( p_2, P_2)\in\llbracket( s_2, S_2)\rrbracket$
  for which $p_2\tto[ a, x_2] q_2$, and
  \begin{align*}
    \bdist\big(( q_1', P_1'),( q_2, P_2)\big) &\le \lambda^{ -1}\big(
    \epsilon- \kdist(( a, x_1),( a, x_2))\big) \\
    &= \lambda^{ -1}\big( \epsilon- \kdist(( a, I_1),( a, I_2))\big),
  \end{align*}
  which, as $( q_1', P_1')\in\llbracket( t_1, S_1)\rrbracket$ was
  chosen arbitrarily, entails $\tdist\big(( s_1, S_1),( s_2,
  S_2)\big)\le \lambda^{ -1}\big( \epsilon- \kdist(( a, I_1),( a,
  I_2))\big)$.

  We are ready to show that $R$ is a refinement family.  Let $( s_1,
  s_2)\in R_\epsilon\in R$ for some $\epsilon$, and assume $s_1\mayto[
  a,I_1]_1 t_1$.  Let $x\in I_1$, then there is an
  implementation $( p, P^x)\in\llbracket( s_1, S_1)\rrbracket$ with a
  transition $p\tto[ m] q$.  Now $\tdist\big(( p, P^x),( s_2,
  S_2)\big)\le \epsilon$, hence we have a transition $s_2\mayto[ a,
  I_2^x]_2 t_2^x$ with $\kdist\big(( a, x),( a, I_2^x)\big)\le \epsilon$.
  Also for any other $x'\in I_1$ we have a transition $s_2\mayto[ a,
  I_2^{ x'}]_2 t_2^{ x'}$ with $\kdist\big(( a, x'),( a, I_2^{
    x'})\big)\le \epsilon$, hence by determinism of $S_2$, $I_2^x=
  I_2^{ x'}$ and $t_2^x= t_2^{ x'}$.  It follows that there is a
  unique transition $s_2\mayto[ a, I_2] t_2$, and as $\kdist\big(( a,
  x),( a, I_2)\big)\le \epsilon$ for all $x\in I_1$, we have $\kdist\big((
  a, I_1),( a, I_2)\big)\le \epsilon$, and $( t_1, t_2)\in R_{
    \lambda^{ -1}( \epsilon- \kdist(( a, I_1),( a, I_2)))}$ by definition.

  Now assume $s_2\mustto[ a,I_2]_2 t_2$. Let $( p_1,
  P_1)\in\llbracket( s_1, S_1)\rrbracket$, then we have $( p_2,
  P_2)\in\llbracket( s_2, S_2)\rrbracket$ with $\bdist\big(( p_1, P_1),(
  p_2, P_2)\big)\le \epsilon$.  Now any $( p_2, P_2)\in\llbracket(
  s_2, S_2)\rrbracket$ has $p_2\tto[ a, x_2] q_2$ with $x_2\in I_2$,
  thus there is also $p_1\tto[ a, x_1] q_1$ with $\kdist\big(( a, x_1),(
  a, x_2)\big)\le \epsilon$ and $\bdist\big(( q_1, P_1),( q_2,
  P_2)\big)\le \lambda^{ -1}\big( \epsilon- \kdist(( a, x_1),( a,
  x_2))\big)$.  This in turn implies that $s_1\mustto[ a, I_1]_1
  t_1$ for some $x_1\in I_1$.  We will be done once we can show
  $\kdist\big(( a, I_1),( a, I_2)\big)\le \epsilon$, so assume to the
  contrary that there is $x_1'\in I_1$ with $\kdist\big(( a, x_1'),( a,
  I_2)\big)> \epsilon$.  Then there is an implementation $( p_1',
  P_1')\in\llbracket( s_1, S_1)\rrbracket$ with $p_1'\tto[ a, x_1']
  q_1'$, hence a transition $s_2\mayto[ a, I_2']_2 t_2'$ with
  $\kdist\big(( a, x_1'),( a, I_2')\big)\le \epsilon$.  But $I_2'= I_2$ by
  determinism of $S_2$, a contradiction. \qed  
\end{proof}

\section{Complexity of Computing Thorough and Modal Refinement Distances}
\label{se:complex}

The complexity results in the next theorem show that modal refinement
distance can serve as a useful approximation of thorough refinement
distance.

\begin{theorem}
  \label{pr:exptime-hard}
  For finite WMTS $S_1$, $S_2$ and $\epsilon \ge 0$, it is
  \EXPTIME-hard to decide whether $S_1\le_t^\epsilon
  S_2$.  The problem whether $S_1\le_m^\epsilon S_2$ is decidable in
  \NP~$\cap$~\coNP.
\end{theorem}

The fact that computing thorough refinement distance is EXPTIME-hard is
easy.  By~\cite{DBLP:conf/ictac/BenesKLS09}, deciding thorough
refinement for MTS (without weights) is EXPTIME-complete.  By
translating MTS to WMTS with weight $0$ on all transitions, deciding
thorough refinement for modal transition systems polynomial-time reduces
to deciding whether thorough refinement distance is $\le 0$.

To show an upper bound on the complexity of computing modal refinement
distance, we need to introduce \emph{discounted values of weighted
  games}, \cf~\cite{DBLP:conf/cocoon/ZwickP95}.  A weighted game graph
is a finite real-weighted bipartite digraph $( V_1, V_2, \mgameto)$,
\ie~with $V_1\cap V_2= \emptyset$ and $\mgameto \in( V_1\times
\Real\times V_2)\cup( V_2\times \Real\times V_1)$ a finite set of
edges.  These are assumed to be non-blocking in the sense that each
$v\in V_1\cup V_2$ has at least one outgoing edge $v\gameto[ r] w$
(which is the shorthand for $( v, r, w)\in \mgameto$).

A Player-1 strategy in such a weighted game graph is a mapping
$\theta_1: V_1\to \Real\times V_2$ for which $\big( v_1, \theta_1(
v_1)\big)\in \mgameto$ for each $v_1\in V_1$.  Similarly, a Player-2
strategy is a mapping $\theta_2: V_2\to \Real\times V_1$ such that
$\big( v_2, \theta_2( v_2)\big)\in \mgameto$ for each $v_2\in V_2$.
The sets of all Player-1 and Player-2 strategies are denoted
$\Theta_1$ and $\Theta_2$, respectively.

Denote by $\tgt( e)= w$ the target of an edge $e=( v, r, w)\in
\mgameto$ and by $\weight( e)= r$ its weight.  A vertex $v_0\in V_1$
and a pair $( \theta_1, \theta_2)\in \Theta_1\times \Theta_2$ of
strategies determine a unique infinite sequence $\big( e_j( \theta_1,
\theta_2)\big)_{ j\ge 0}$ of edges $e_j( \theta_1, \theta_2)\in
\mgameto$ for which
\begin{align*}
  e_0( \theta_1, \theta_2) &= \big( v_0, \theta_1( v_0)\big), \\
  e_{ 2j+ 1}( \theta_1, \theta_2) &= \big( \tgt( e_{ 2j}), \theta_2(
  \tgt( e_{ 2j}))\big), \\
  e_{ 2j}( \theta_1, \theta_2) &= \big( \tgt( e_{ 2j- 1}, \theta_1(
  \tgt( e_{ 2j- 1}))\big).
\end{align*}
In other words, the two players alternate to pick edges in $\mgameto$
according to their strategies.  The \emph{discounted value} of the
game $( V_1, V_2, \mgameto)$ played from $v_0\in V_1$ with discounting
factor $\lambda$, $0\le \lambda< 1$, is defined to be
\begin{equation*}
  p( v_0, \lambda)= \adjustlimits \sup_{ \theta_1\in \Theta_1} \inf_{
    \theta_2\in \Theta_2} \sum_{ j= 0}^\infty \lambda^j \weight\big(
  e_j( \theta_1, \theta_2)\big).
\end{equation*}

We recall the following theorem from~\cite{DBLP:conf/cocoon/ZwickP95};
the complexity result is obtained by reduction to simple stochastic
games~\cite{DBLP:journals/iandc/Condon92}.

\begin{lemma}[\cite{DBLP:conf/cocoon/ZwickP95}]
  The discounted value $p( v_0, \lambda)$ may be computed as the
  unique fixed point to the equations
  \begin{equation*}
    p( v, \lambda)=
    \begin{cases}
      \max\limits_{ v\gameto[ r] w} r+ \lambda p( w, \lambda)
      &\text{if } v\in V_1, \\
      \min\limits_{ v\gameto[ r] w} r+ \lambda p( w, \lambda)
      &\text{if } v\in V_2.
    \end{cases}
  \end{equation*}
  The decision problem corresponding to computing $p( v_0)$ is
  contained in \NP~$\cap$~\coNP.
\end{lemma}

Next we present a reduction from modal refinement distance of WMTS to
discounted values of weighted games,
\cf~\cite{journals/tcs/LarsenFT11}.

\begin{lemma}
  For WMTS $S_1$, $S_2$ one can construct in polynomial time a
  weighted game $( V_1, V_2, \mgameto)$ with a vertex $v_0\in V_1$
  such that $\mdist( S_1, S_2)= p( v_0, \sqrt{ \lambda})$.
\end{lemma}

\begin{proof}
  Let $V_1= S_1\times S_2$, $V_2= S_1\times S_2\times \K\times\{ \may,
  \must\}$, and define the transitions as follows:
  \begin{alignat*}{2}
    ( s_1, s_2) &\gameto[ 0] ( t_1, s_2, k_1, \may) &
    \quad\text{if}\quad s_1 &\mayto[ k_1]_1 t_1 \\
    ( s_1, s_2) &\gameto[ 0] ( s_1, t_2, k_2, \must) &
    \quad\text{if}\quad s_2 &\mustto[ k_2]_2 t_2 \\
    ( t_1, s_2, k_1, \may) &\gameto[ \kdist( k_1, k_2)] ( t_1, t_2) &
    \quad\text{if}\quad s_2 &\mayto[ k_2]_2 t_2 \\
    ( s_1, t_2, k_2, \must) &\gameto[ \kdist( k_1, k_2)] ( t_1, t_2) &
    \quad\text{if}\quad s_1 &\mustto[ k_1]_1 t_1
  \end{alignat*}
  Setting $v_0=( s_1^0, s_2^0)$ finishes the construction. \qed
\end{proof}

In~\cite{journals/tcs/LarsenFT11} it is also shown that conversely,
computing discounted values of weighted games may be polynomial-time
reduced to computing simulation distance for weighted transition
systems, hence we can conclude the following.

\begin{lemma}
  The decision problem corresponding to computing modal refinement
  distance for WMTS is polynomial-time equivalent to the decision
  problem corresponding to computing discounted values of weighted
  games.
\end{lemma}

\section{Relaxation}
\label{se:relax}

We introduce here a notion of \emph{relaxation} which is specific to the
quantitative setting.  Intuitively, relaxing a specification means to
weaken the quantitative constraints, while the discrete demands on which
transitions may or must be present in implementations are kept.  A
similar notion of \emph{strengthening} may be defined, but we do not use
this here.

\begin{definition}
  \label{def:relax}
  For WMTS $S$, $S'$ and $\epsilon\ge 0$, $S'$ is an
  \emph{$\epsilon$-relaxation} of $S$ if $S\le_m S'$ and
  $S'\le_m^\epsilon S$.
\end{definition}

Hence the quantitative constraints in $S'$ may be more permissive than
the ones in $S$, but no new discrete behavior may be introduced.  Also
note that any implementation of $S$ is also an implementation of $S'$,
and no implementation of $S'$ is further than $\epsilon$ away from an
implementation of $S$.  The following proposition relates specifications
to relaxed specifications:

\begin{proposition}
  \label{pr:wide-two}
  If $S_1'$ and $S_2'$ are $\epsilon$-relaxations of $S_1$ and $S_2$,
  respectively, then $\mdist( S_1, S_2)- \epsilon\le \mdist( S_1,
  S_2')\le \mdist( S_1, S_2)$ and $\mdist( S_1, S_2)\le \mdist( S_1',
  S_2)\le \mdist( S_1, S_2)+ \epsilon$.
\end{proposition}

\begin{proof}
  By the triangle inequality we have
  \begin{align*}
    \mdist( S_1, S_2') &\le \mdist( S_1, S_2)+ \mdist( S_2, S_2'), \\
    \mdist( S_1, S_2) &\le \mdist( S_1, S_2')+ \mdist( S_2', S_2), \\
    \mdist( S_1, S_2) &\le \mdist( S_1, S_1')+ \mdist( S_1', S_2), \\
    \mdist( S_1', S_2) &\le \mdist( S_1', S_1)+ \mdist( S_1,
    S_2). \qed
  \end{align*}
\end{proof}

On the syntactic level, we can introduce the following \emph{widening}
operator which relaxes all quantitative constraints in a systematic
manner.  We write $I\pm \delta=[ x- \delta, y+ \delta]$ for an interval
$I=[ x, y]$ and $\delta\in \Nat$.

\begin{definition}\label{def:widening}
  Given $\delta\in \Nat$, the \emph{$\delta$-widening} of a WMTS $S$
  is the WMTS $S^{ +\delta}$ with transitions $s\mayto[ a, I\pm
  \delta] t$ in $S^{ +\delta}$ for all $s\mayto[ a, I] t$ in $S$, and
  $s\mustto[ a, I\pm \delta] t$ in $S^{ +\delta}$ for all $s\mustto[
  a, I] t$ in $S$.
\end{definition}

Widening and relaxation are related as follows; note also that as
widening is a global operation whereas relaxation may be achieved
entirely locally, not all relaxations may be obtained as widenings.

\begin{proposition}
  \label{pr:wide-propt}
  The $\delta$-widening of any WMTS $S$ is a $( 1- \lambda)^{ -1}
  \delta$-relaxation.
\end{proposition}

\begin{proof}
  For the first claim, the identity relation $\id_S=\{( s, s)\mid s\in
  S\}\subseteq S\times S$ is a witness for $S\le_m S^{ +\delta}$: if
  $s\mayto[ k] t$, then by construction $s\mayto[ k_2]_{ +\delta} t$
  with $k\sqsubseteq k_2$, and if $s\mustto[ k_2]_{ +\delta} t$, then
  again by construction $s\mustto[ k] t$ for some $k\sqsubseteq k_2$.

  Now to prove $\mdist( S^{ +\delta}, S)\le( 1- \lambda)^{ -1}
  \delta$, we define a family of relations $R=\{ R_\epsilon\mid
  \epsilon\ge 0\}$ by $R_\epsilon= \emptyset$ for $\epsilon<( 1-
  \lambda)^{ -1} \delta$ and $R_\epsilon= \id_S$ for $\epsilon\ge {(
    1- \lambda)^{ -1} \delta}$.  We show that $R$ is a modal
  refinement family.

  Let $( s, s)\in R_\epsilon$ for some $\epsilon\ge {( 1- \lambda)^{
      -1} \delta}$, and assume $s\mayto[ k_2]_{ +\delta} t$.  By
  construction there is a transition $s\mayto[ k] t$ with $\kdist( k_2,
  k)\le \delta\le \epsilon$.  Now
  \begin{equation*}
    \frac1\lambda\Big( \epsilon- \kdist( k_2, k)\Big)\ge \frac1\lambda\Big(
    \frac \delta{ 1- \lambda}- \delta\Big)= \frac \delta{ 1-
      \lambda}\ge \epsilon
  \end{equation*}
  and $( t, t)\in R_\epsilon$, which settles this part of the proof.
  The other direction, starting with a transition $s\mustto[ k] t$, is
  similar.  \qed
\end{proof}

There is also an implementation-level notion which corresponds to
relaxation:

\begin{definition}
  The \emph{$\epsilon$-extended implementation semantics}, for
  $\epsilon\ge 0$, of a WMTS $S$ is $\llbracket S\rrbracket^{
    +\epsilon}=\big\{ I\bigmid I\le_m^\epsilon S, I \text{
    implementation}\big\}$.
\end{definition}

\begin{proposition}
  \label{pr:semwide}
  If $S'$ is an $\epsilon$-relaxation of $S$, then $\llbracket
  S'\rrbracket\subseteq \llbracket S\rrbracket^{ +\epsilon}$.
\end{proposition}

\begin{proof}
  If $I\in \llbracket S'\rrbracket$, then $\mdist( I, S')= 0$, hence
  $\mdist( I, S)\le \epsilon$ by Proposition~\ref{pr:wide-two}, which in
  turn implies that $I\in \llbracket S\rrbracket^{+\epsilon}$. \qed
\end{proof}

The example in Figure~\ref{fi:sem-wide-counterex} shows that there are
WMTS $S$, $S'$ such that $S'$ is an $\epsilon$-relaxation of $S$ but the
inclusion $\llbracket S'\rrbracket\subseteq \llbracket S\rrbracket^{
  +\epsilon}$ is strict.  Indeed, for $\delta= 1$ and $\lambda= .9$, we
have $I\in \llbracket S\rrbracket^{ +( 1- \lambda)^{ -1} \delta}$, but
$I\notin \llbracket S^{ +\delta}\rrbracket$.

\begin{figure}[tpb]
  \centering
  \subfigure[$S$]{
    \begin{tikzpicture}[->,>=stealth',shorten >=1pt,auto,node
      distance=2.0cm,initial text=,scale=0.8,transform shape]
      \tikzstyle{every node}=[font=\small] \tikzstyle{every
        state}=[fill=white,shape=circle,inner sep=.5mm,minimum size=6mm]
      \node[initial,state] (i) at (0,0) {$s$};
      \node[state] (j) at (2,0) {$t$};
      \node[state] (k) at (4,0) {$u$};
      \path (i) edge node [above] {$a,[ 5, 5]$} (j);
      \path (j) edge node [above] {$a,[ 5, 5]$} (k);
    \end{tikzpicture}}
  \quad
  \subfigure[$S^{ +1}$]{
    \begin{tikzpicture}[->,>=stealth',shorten >=1pt,auto,node
      distance=2.0cm,initial text=,scale=0.8,transform shape]
      \tikzstyle{every node}=[font=\small] \tikzstyle{every
        state}=[fill=white,shape=circle,inner sep=.5mm,minimum size=6mm]
      \node[initial,state] (i) at (0,0) {$s^{ +1}$};
      \node[state] (j) at (2,0) {$t^{ +1}$};
      \node[state] (k) at (4,0) {$u^{ +1}$};
      \path (i) edge node [above] {$a,[ 4, 6]$} (j);
      \path (j) edge node [above] {$a,[ 4, 6]$} (k);
    \end{tikzpicture}}
  \quad
  \subfigure[$I$]{
    \begin{tikzpicture}[->,>=stealth',shorten >=1pt,auto,node
      distance=2.0cm,initial text=,scale=0.8,transform shape]
      \tikzstyle{every node}=[font=\small] \tikzstyle{every
        state}=[fill=white,shape=circle,inner sep=.5mm,minimum size=6mm]
      \node[initial,state] (i) at (0,0) {$i$};
      \node[state] (j) at (2,0) {$j$};
      \node[state] (k) at (4,0) {$k$};
      \path (i) edge node [above] {$a, 15$} (j);
      \path (j) edge node [above] {$a, 5$} (k);
    \end{tikzpicture}}
  \caption{WMTS $S$ and implementation $I$ for which $I\in \llbracket
    S\rrbracket^{ +( 1- \lambda)^{ -1} \delta}$, for $\delta= 1$ and
    $\lambda= .9$ (thus $( 1- \lambda)^{ -1} \delta= 10$), but $I\notin
    \llbracket S^{ +\delta}\rrbracket$, so that $\llbracket S^{
      +\delta}\rrbracket\subsetneq \llbracket
    S\rrbracket^{ +( 1- \lambda)^{ -1} \delta}$, even though $S^{
      +\delta}$ is a $( 1- \lambda)^{ -1} \delta$-relaxation of $S$.}
  \label{fi:sem-wide-counterex}
\end{figure}
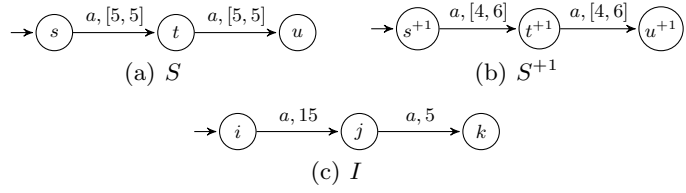

\section{Limitations of the Quantitative Approach}
\label{se:hull}

In this section we turn our attention towards some of the standard
operators for specification theories; determinization and logical
conjunction. In the standard Boolean setting, there is indeed a
determinization operator which derives the smallest deterministic
overapproximation of a specification, which is useful because it
enables checking thorough refinement,
\cf~Theorem~\ref{th:det-dteqdm}. Quite surprisingly, we show that in
the quantitative setting, there are problems with these notions which
do not appear in the Boolean theory.  More specifically, we show that
there is no determinization operator which always yields a
smallest deterministic overapproximation, and there is no
conjunction operator which acts as a greatest lower bound.

\begin{theorem}
  \label{th:no-dethull}
  There is no unary operator $\mathcal D$ on WMTS for which it holds
  that
  \begin{enumerate}[$(\ref{th:no-dethull}.1)$]
  \item \label{en:no-dethull:det}
    $\mathcal D( S)$ is deterministic for any WMTS $S$,
  \item
    \label{en:no-dethull:ub}
    $S\le_m \mathcal D( S)$ for any WMTS $S$,
  \item
    \label{en:no-dethull:lub}
    $S\le_m^\epsilon D$ implies $\mathcal D( S)\le_m^\epsilon D$ for
    any WMTS $S$, any deterministic WMTS $D$, and any $\epsilon\ge 0$.
  \end{enumerate}
\end{theorem}

\begin{proof}
  There is a determinization operator $\mathcal D'$ on WMTS which
  satisfies Properties~$(\ref{th:no-dethull}.\ref{en:no-dethull:det})$
  and $(\ref{th:no-dethull}.\ref{en:no-dethull:ub})$ above and a weaker
  version of Property~$(\ref{th:no-dethull}.\ref{en:no-dethull:lub})$
  with $\epsilon= 0$:
  \begin{enumerate}[$(\ref{th:no-dethull}.1')$]
  \item[$(\ref{th:no-dethull}.\ref{en:no-dethull:lub}')$] $S\le_m D$
    implies $\mathcal D'( S)\le_m D$ for any WMTS $S$ and any
    deterministic WMTS $D$.
  \end{enumerate}
  This $\mathcal D'$ can be defined as follows: For a WMTS $S =
  (S,s_0,\mmayto,\mmustto)$,
  \begin{equation*}
    \mathcal D'(S) = \big( \powerset{S} \setminus
    \{\emptyset\},\{s_0\},\mmayto_d,\mmustto_d\big),
  \end{equation*}
  where $\powerset S$ is the power set of $S$ and the transition
  relations $\mmayto_d$ and $\mmustto_d$ are defined as follows: Let
  $\mathcal T \in (\powerset{S} \setminus \{\emptyset\})$ be a state in
  $\mathcal D'(S)$. For every maximal, nonempty set $L_a \subseteq \{ I
  \mid \exists s \in \mathcal T : s \mayto[a,I] \}$ for some $a \in
  \Sigma$, we have $\mathcal T \mayto[a,\bigcup L_a]_d \mathcal T_a$
  where $\mathcal T_a = \{ s' \in S \mid \exists s \in \mathcal T, I \in
  L_a : s\mayto [ a,I] s' \}$ and $\bigcup L_a$ is the smallest interval
  containing all intervals from $L_a$. If, moreover, for each $s \in
  \mathcal T$ we have $s \mustto[a,I] s'$ for some $s' \in \mathcal T_a$
  and some $I \in L_a$, then $\mathcal T \mustto[ a,\bigcup L_a ]_d
  \mathcal T_a$. It is straightforward to prove that $\mathcal D'$
  satisfies the expected properties.

  Assume now that there is an operator $\mathcal D$ as in the theorem.
  Then for any WMTS $S$, $S\le_m \mathcal D'( S)$ and thus $\mathcal D(
  S)\le_m \mathcal D'( S)$
  by~$(\ref{th:no-dethull}.\ref{en:no-dethull:lub})$, and $S\le_m
  \mathcal D( S)$ and hence $\mathcal D'( S)\le_m \mathcal D( S)$
  by~$(\ref{th:no-dethull}.\ref{en:no-dethull:lub}')$.  We finish the
  proof by showing that the operator $\mathcal D'$ does not
  satisfy~$(\ref{th:no-dethull}.\ref{en:no-dethull:lub})$.  The example
  in Figure~\ref{fig:cex} shows a WMTS $S$ and a deterministic WMTS $D$
  for which $\mdist\big( \mathcal D'( S), D\big)= 3+ 3\lambda$ and
  $\mdist( S, D)= \max( 3, 3\lambda)= 3$, hence $\mdist\big( \mathcal
  D'( S), D\big)\not\le \mdist( S, D)$. \qed
\end{proof}

Likewise, the greatest-lower-bound property of logical conjunction in
the Boolean setting ensures that the set of implementations of a
conjunction of specifications is precisely the intersection of the
implementation sets of the two specifications.  Conjoining two WMTS
naturally involves a partial label conjunction operator $\conj$. We let
$( a_1, I_1)\conj( a_2, I_2)$ be undefined if $a_1\ne a_2 $, and
otherwise
\begin{align*}
  \big( a,[ x_1, y_1]\big)\conj\big( a,[ x_2, y_2]\big) &= 
  \begin{cases}
    \big( a,[ \max(x_1,x_2), \min(y_1, y_2)]\big) \\
    &\hspace*{-5em}\text{if }
    \max(x_1,x_2) \le \min(y_1, y_2),\\
    \text{undefined} &\hspace*{-5em}\text{otherwise}.
\end{cases}
\end{align*}

Before we show that such a conjunction operator for WMTS does not exist
in general, we need to define a \emph{pruning operator} which removes
inconsistent states that naturally arise when conjoining two WMTS.  The
intuition is that if a WMTS $S_1$ requires a behavior $s_1 \mustto[ k_1
]_1$ for which there is no may transition $s_2 \mayto[ k_2]_2$ such that
$k_1 \conj k_2$ is defined, then the state $(s_1,s_2)$ in the
conjunction is \emph{inconsistent} and will have to be pruned away,
together with all \must transitions leading to it.  In the definition
below, $\pre^*$ denotes the reflexive, transitive closure of $\pre$.

\begin{definition}
  \label{de:prune}
  For a WMTS $S$, let $\pre: 2^S\to 2^S$ be given by $\pre( B)=\{ s\in
  S\mid s\mustto[ k] t\in B \text{ for some } k\}$. Let $\lightning
  \subseteq S$ be a set of \emph{inconsistent} states. If $s^0\notin
  \pre^*( \lightning )$, then the \emph{pruning of $S$
    w.r.t.~$\lightning$} is defined by $\rho^\lightning( S)=( S_\rho,
  s^0, \mmayto_\rho, \mmustto_\rho)$ where $S_\rho= S\setminus \pre^*(
  \lightning )$, $\mmayto_\rho= \mmayto\cap\big( S_\rho\times \K \times
  S_\rho\big)$ and $\mmustto_\rho= \mmustto\cap\big( S_\rho\times \K
  \times S_\rho\big)$.
\end{definition}

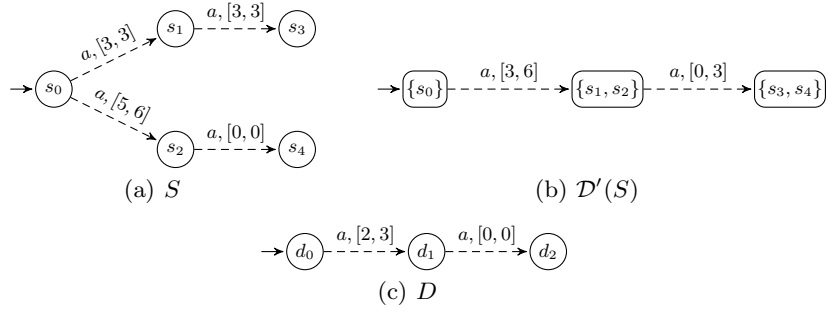
\begin{figure}[tpb]
  \centering 
  \subfigure[$S$]{
    \begin{tikzpicture}[->,>=stealth',shorten >=1pt,auto,node
      distance=2.0cm,initial text=,scale=0.8,transform shape]
      \tikzstyle{every node}=[font=\small] \tikzstyle{every
        state}=[fill=white,shape=circle,inner sep=.5mm,minimum size=6mm]
      \node[initial,state] (i) at (0,0) {$s_0$};
      \node[state] (s) at (2,1) {$s_1$};
      \node[state] (t) at (2,-1) {$s_2$};
      \node[state] (spost) at (4,1) {$s_3$};
      \node[state] (tpost) at (4,-1) {$s_4$};
      \path (i) edge [densely dashed] node [above,sloped] {$a, [3,3]$} (s);
      \path (i) edge [densely dashed] node [above,sloped] {$a, [5,6]$} (t);
      \path (t) edge [densely dashed] node [above,sloped] {$a,
        [0,0]$} (tpost); 
      \path (s) edge [densely dashed] node [above,sloped] {$a,
        [3,3]$} (spost); 
    \end{tikzpicture}}
  \hspace{2mm}
  \subfigure[$\mathcal D'(S)$]{
    \raisebox{8mm}{
      \begin{tikzpicture}[->,>=stealth',shorten >=1pt,auto,node
        distance=2.0cm,initial text=,scale=0.8,transform shape]
        \tikzstyle{every node}=[font=\small] \tikzstyle{every
          state}=[fill=white,shape=circle,inner sep=.5mm,minimum size=6mm]
        \node[initial,state,shape=rectangle,rounded corners] (i) at
        (0,0) {$\{s_0\}$}; 
        \node[state,shape=rectangle,rounded corners] (s) at (3,0)
        {$\{s_1,s_2\}$}; 
        \node[state,shape=rectangle,rounded corners] (t) at (6,0)
        {$\{s_3,s_4\}$}; 
        \path (i) edge [densely dashed] node [above,sloped] {$a,
          [3,6]$} (s); 
        \path (s) edge [densely dashed] node [above,sloped] {$a,
          [0,3]$} (t); 
      \end{tikzpicture}}}
  \hspace{2mm}
  \subfigure[$D$]{
    \begin{tikzpicture}[->,>=stealth',shorten >=1pt,auto,node
      distance=2.0cm,initial text=,scale=0.8,transform shape]
      \tikzstyle{every node}=[font=\small] \tikzstyle{every
        state}=[fill=white,shape=circle,inner sep=.5mm,minimum size=6mm]
      \node[initial,state] (i) at (0,0) {$d_0$};
      \node[state] (s) at (2,0) {$d_1$};
      \node[state] (t) at (4,0) {$d_2$};
      \path (i) edge [densely dashed] node [above,sloped] {$a, [2,3]$} (s);
      \path (s) edge [densely dashed] node [above,sloped] {$a, [0,0]$} (t);
    \end{tikzpicture}
  }
  \caption{Counter-example for Theorem~\ref{th:no-dethull}: $\mdist\big(
    \mathcal D'( S), D\big)= 3+ 3\lambda$ and $\mdist( S, D)= \max( 3,
    3\lambda)= 3$, hence $\mdist\big( \mathcal D'( S), D\big)\not\le
    \mdist( S, D)$.}
  \label{fig:cex}
\end{figure}

\begin{theorem}
  \label{th:no-conj}
  There is no partial binary operator $\wedge$ on WMTS for which it
  holds that, for all WMTS $S$, $S_1$, $S_2$ such that $S_1$ and $S_2$
  are deterministic,
  \begin{enumerate}[$(\ref{th:no-conj}.1)$]
  \item
    \label{en:no-conj:lb}
    whenever $S_1 \wedge S_2$ is defined, then
    $S_1\wedge S_2\le_m S_1$ and $S_1\wedge S_2\le_m S_2$,
  \item 
    \label{en:no-conj:def}
    whenever $S\le_m S_1$ and
    $S\le_m S_2$, then $S_1 \wedge S_2$ is defined and $S\le_m S_1\wedge S_2$,
  \item
    \label{en:no-conj:glb}
    for any $\epsilon\ge 0$, there exist $\epsilon_1\ge 0$ and
    $\epsilon_2\ge 0$ such that if $S_1 \wedge S_2$ is defined,
    $S\le_m^{ \epsilon_1} S_1$ and $S\le_m^{ \epsilon_2} S_2$, then
    $S\le_m^\epsilon S_1\wedge S_2$.
  \end{enumerate}
\end{theorem}

\begin{proof}
  We follow the same strategy as in the proof of
  Theorem~\ref{th:no-dethull}. One can define a partial conjunction
  operator $\wedge'$ defined for WMTS which satisfies
  Properties~$(\ref{th:no-conj}.\ref{en:no-conj:lb})$
  and~$(\ref{th:no-conj}.\ref{en:no-conj:def})$ as follows: For
  deterministic WMTS $S_1$ and $S_2$, $S_1 \wedge' S_2 =
  \rho^\lightning(S_1 \times S_2,(s_1^0,s_2^0),\mmayto,\mmustto)$ where
  the transition relations $\mmayto$ and $\mmustto$ and the set
  $\lightning \subseteq S_1 \times S_2$ of inconsistent states are
  defined by the following rules:
  \begin{gather*}
    \frac{s_1 \mustto[k_1] s_1'\quad s_2 \mayto[k_2] s_2' \quad k_1
      \conj k_2 \text{ defined}}{(s,t)\mustto[k_1 \conj k_2]
      (s_1',s_2')}%
    \qquad \frac{s_1 \mayto[ k_1 ] s_1'\quad s_2 \mustto[k_2 ] s_2'
      \quad k_1 \conj k_2 \text{
        defined}}{(s_1,s_2)\mustto[ k_1 \conj k_2](s_1',s_2')} \\
    \frac{s_1 \mayto[ k_1 ] s_1'\quad s_2 \mayto[ k_2] s_2' \quad k_1
      \conj k_2 \text{ defined}}{(s_1,s_2)\mayto[k_1 \conj
      k_2](s_1',s_2')} \\
    \frac{s_1 \mustto[k_1] \quad \big(k_1\conj k_2 \text{ undefined for
        any $k_2$ such that } s_2\mayto[k_2]\big)} {(s_1,s_2)\in
      \lightning} \\
    \frac{s_2\mustto[k_2] \quad \big(k_1\conj k_2 \text{ undefined for
        any $k_1$ such that } s_1\mayto[k_1] \big)}{(s_1,s_2) \in
      \lightning}
  \end{gather*}
  
  Using these properties, one can see that for all deterministic WMTS
  $S_1$ and $S_2$, $S_1\wedge S_2\le_m S_1\wedge' S_2$ and $S_1\wedge'
  S_2\le_m S_1\wedge S_2$.  The WMTS depicted in
  Figure~\ref{fi:conjunct-counterex} then show that
  Property~$(\ref{th:no-conj}.\ref{en:no-conj:glb})$ cannot hold: here,
  $\mdist( S, S_1)= \mdist( S, S_2)= 1$, but $\mdist( S, S_1\wedge S_2)=
  \infty$. \qed
\end{proof}

\begin{figure}[tpb]
  \centering
  \subfigure[$S$]{
    \begin{tikzpicture}[->,>=stealth',shorten >=1pt,auto,node
      distance=2.0cm,initial text=,scale=0.8,transform shape]
      \tikzstyle{every node}=[font=\small] \tikzstyle{every
        state}=[fill=white,shape=circle,inner sep=.5mm,minimum size=6mm]
      \node[initial,state] (s) at (0,0) {$s$};
      \node[state] (t) at (2,0) {$t$};
      \path (s) edge [densely dashed] node [above] {$a,[ 1, 2]$} (t);
    \end{tikzpicture}}
  \qquad
  \subfigure[$S_1$]{
    \begin{tikzpicture}[->,>=stealth',shorten >=1pt,auto,node
      distance=2.0cm,initial text=,scale=0.8,transform shape]
      \tikzstyle{every node}=[font=\small] \tikzstyle{every
        state}=[fill=white,shape=circle,inner sep=.5mm,minimum size=6mm]
      \node[initial,state] (s) at (0,0) {$s_1$};
      \node[state] (t) at (2,0) {$t_1$};
      \path (s) edge [densely dashed] node [above] {$a,[ 0, 1]$} (t);
    \end{tikzpicture}}
  \\
  \subfigure[$S_2$]{
    \begin{tikzpicture}[->,>=stealth',shorten >=1pt,auto,node
      distance=2.0cm,initial text=,scale=0.8,transform shape]
      \tikzstyle{every node}=[font=\small] \tikzstyle{every
        state}=[fill=white,shape=circle,inner sep=.5mm,minimum size=6mm]
      \node[initial,state] (s) at (0,0) {$s_2$};
      \node[state] (t) at (2,0) {$t_2$};
      \path (s) edge [densely dashed] node [above] {$a,[ 2, 3]$} (t);
    \end{tikzpicture}}
  \qquad
  \subfigure[$S_1\wedge S_2$]{
    \begin{tikzpicture}[->,>=stealth',shorten >=1pt,auto,node
      distance=2.0cm,initial text=,scale=0.8,transform shape]
      \tikzstyle{every node}=[font=\small] \tikzstyle{every
        state}=[fill=white,shape=circle,inner sep=.5mm,minimum size=6mm]
      \node[initial,state,shape=rectangle,rounded corners] (s) at
      (0,0) {$( s_1, s_2)$};
    \end{tikzpicture}}
  \caption{Counter-example for Theorem~\ref{th:no-conj}: $\mdist( S,
    S_1)= \mdist( S, S_2)= 1$, but $\mdist( S, S_1\wedge S_2)= \infty$.}
  \label{fi:conjunct-counterex}
\end{figure}
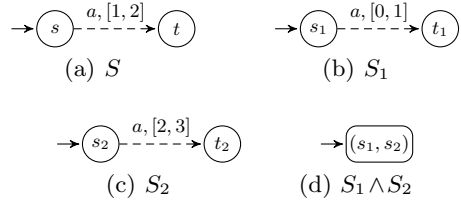

The counterexamples used in the proofs of Theorems~\ref{th:no-dethull}
and~\ref{th:no-conj} are quite general and apply to a large class of
distances, rather than only to the accumulating distance discussed in
this paper.  Hence it can be argued that what we have exposed here is a
fundamental limitation of any quantitative approach to modal
specifications.

\section{Structural Composition and Quotient}
\label{se:parcomp}

In this section we show that in our quantitative setting, notions of
structural composition and quotient can be defined which obey the
properties expected of such operations.  In particular, structural
composition satisfies independent implementability~\cite{AlfaroH95},
hence the refinement distance between structural composites can be
bounded by the distances between their respective components.

First we define partial synchronization operators $\oplus$ and $\ominus$
on specification labels which will be used for synchronizing
transitions.  We let $( a_1, I_1)\oplus( a_2, I_2)$ and $( a_1,
I_1)\ominus( a_2, I_2)$ be undefined if $a_1\ne a_2 $, and otherwise
\begin{align*}
  \big( a,[ x_1, y_1]\big)\oplus\big( a,[ x_2, y_2]\big) &= \big( a,[
  x_1+ x_2, y_1+ y_2]\big), \\
  \big( a,[ x_1, y_1]\big)\ominus\big( a,[ x_2, y_2]\big) &= \left\{
    \begin{array}{ll}
      \text{undefined} &\quad\text{if } x_1- x_2> y_1- y_2, \\
      \big( a,[ x_1- x_2, y_1- y_2]\big) &\quad\text{if } x_1- x_2\le
      y_1- y_2.
    \end{array}
  \right.
\end{align*}
Note that we use CSP-style synchronization, but other types of
synchronization can easily be defined.  Also, defining $\oplus$ to add
intervals (and $\ominus$ to subtract them) is only one particular
choice; depending on the application, one can also \eg~let $\oplus$ be
intersection of intervals or some other operation.  It is not difficult
to see that these alternative synchronization operators would lead to
properties similar to those we show here.

\begin{definition}
  \label{de:comp-quot}
  Let $S_1$ and $S_2$ be WMTS.  The \emph{structural composition} of $S_1$
  and $S_2$ is $S_1\| S_2=\big( S_1\times S_2,( s_1^0, s_2^0), \K,
  \mmayto, \mmustto\big)$ with transitions given as follows:
    \begin{gather*}
      \dfrac{ s_1\mayto[ k_1]_1 t_1 \quad s_2\mayto[ k_2]_2 t_2 \quad
        k_1\oplus k_2 \text{ def.}}{( s_1, s_2)\mayto[ k_1\oplus
        k_2]( t_1, t_2)}%
      \qquad \dfrac{ s_1\mustto[ k_1]_1 t_1 \quad s_2\mustto[ k_2]_2 t_2
        \quad k_1\oplus k_2 \text{ def.}} {( s_1, s_2)\mustto[
        k_1\oplus k_2]( t_1, t_2)}
    \end{gather*}
    The \emph{quotient} of $S_1$ by $S_2$ is $S_1\bbslash S_2=
    \rho^{\lightning}\big( S_1\times S_2\cup\{ u\},( s_1^0, s_2^0), \K,
    \mmayto, \mmustto\big)$ with transitions and the set of inconsistent
    states given as follows:
    \begin{gather*}
    \dfrac{%
      s_1\mayto[ k_1]_1 t_1 \quad s_2\mayto[ k_2]_2 t_2 \quad
      k_1\ominus k_2 \text{ def.}}{%
      ( s_1, s_2)\mayto[ k_1\ominus k_2]( t_1, t_2)} \qquad
    \dfrac{%
      s_1\mustto[ k_1]_1 t_1 \quad s_2\mustto[ k_2]_2 t_2 \quad
      k_1\ominus k_2 \text{ def.}}{%
      ( s_1, s_2)\mustto[ k_1\ominus k_2]( t_1, t_2)} \\
    \dfrac{%
      s_1\mustto[ k_1]_1 t_1 \quad \forall s_2\mustto[ k_2]_2 t_2:
      k_1\ominus k_2 \text{ undef.}}{%
      ( s_1, s_2) \in \lightning} \\
    \dfrac{%
      k\in \K \quad \forall s_2\mayto[ k_2]_2 t_2: k\oplus k_2
      \text{ undef.}}{%
      ( s_1, s_2)\mayto[ k] u} \qquad \dfrac{%
      k\in \K}{%
      u\mayto[ k] u}
  \end{gather*}
\end{definition}

Note that during the quotient $S_1\bbslash S_2$ inconsistent states can
arise which are then recursively removed using the pruning operator
$\rho$, see Definition~\ref{de:prune}.  After a technical lemma, the
next theorem shows that structural composition is well-behaved with
respect to modal refinement distance in the sense that the distance
between the composed systems is bounded by the distances of the
individual systems.  Note also the special case in the theorem of
$S_1\le_m S_2$ and $S_3\le_m S_4$ implying $S_1\| S_3\le_m S_2\| S_4$.

\begin{lemma}
  \label{le:synch-cont}
  For $k_1, k_2, k_3, k_4\in \K$ with $k_1\oplus k_3$ and $k_2\oplus
  k_4$ defined, we have $\kdist( k_1\oplus k_3, k_2\oplus k_4)\le
  \kdist( k_1, k_2)+ \kdist( k_3, k_4)$.
\end{lemma}

\begin{proof}
  Let 
  $k_i=\big( a,[ x_i, y_i]\big)$ for all $i$.  We have
  \begin{align*}
    \kdist( k_1, k_2)+ \kdist( k_3, k_4) &= \max( x_2 - x_1, y_1 -
    y_2,0)+ \max( x_4 - x_3, y_3 - y_4,0) \\
    &\ge \max\big(( x_2 - x_1)+( x_4 - x_3),(
    y_1 - y_2)+( y_3 - y_4), 0\big)\\
    &= \max\big(( x_2+ x_4)-(
    x_1+ x_3),( y_1+ y_3)-( y_2+ y_4),0\big) \\
    &= \kdist( k_1\oplus k_3, k_2\oplus k_4). \quad \qed
  \end{align*}
\end{proof}

\begin{theorem}[Independent implementability]
  \label{th:indepimp}
  For WMTS $S_1$, $S_2$, $S_3$, $S_4$ we have $\mdist\big( S_1\| S_3,
  S_2\| S_4\big)\le \mdist( S_1, S_2)+ \mdist( S_3, S_4)$.
\end{theorem}

\begin{proof}
  If $\mdist( S_1, S_2)= \infty$ or $\mdist( S_3, S_4)= \infty$, we have
  nothing to prove.  Otherwise, let $R^1=\{ R^1_\epsilon\subseteq
  S_1\times S_2\mid \epsilon\ge 0\}$, $R^2=\{ R^2_\epsilon\subseteq
  S_3\times S_4\mid \epsilon\ge 0\}$ be witnesses for $\mdist( S_1,
  S_2)$ and $\mdist( S_3, S_4)$, respectively; hence $( s_1^0,
  s_2^0)\in R^1_{ \mdist( S_1, S_2)}\in R^1$ and $( s_3^0, s_4^0)\in
  R^2_{ \mdist( S_3, S_4)}\in R^2$.  Define
  \begin{multline*}
    R_\epsilon=\big\{\big(( s_1, s_3),( s_2, s_4)\big)\in S_1\times
    S_3\times S_2\times S_4\bigmid \\
    ( s_1, s_2)\in R^1_{ \epsilon_1}\in R^1,( s_3, s_4)\in R^2_{
      \epsilon_2}\in R^2, \epsilon_1+ \epsilon_2\le\epsilon \big\}
  \end{multline*}
  for all $\epsilon\ge 0$ and let $R=\{ R_\epsilon\mid \epsilon\ge
  0\}$.  We show that $R$ witnesses $\mdist\big( S_1\| S_3,
  S_2\| S_4\big)\le \mdist( S_1, S_2)+ \mdist( S_3, S_4).$

  We have $\big(( s_1^0, s_3^0),( s_2^0, s_4^0)\big)\in R_{ \mdist(
    S_1, S_2)+ \mdist( S_3, S_4)}\in R$.  Now let $$\big(( s_1, s_3),(
  s_2, s_4)\big)\in R_\epsilon\in R$$ for some $\epsilon$, then $( s_1,
  s_2)\in R^1_{ \epsilon_1}\in R^1$ and $( s_3, s_4)\in R^2_{
    \epsilon_2}\in R^2$ for some $\epsilon_1+ \epsilon_2\le \epsilon$.

  Assume $( s_1, s_3)\mayto[ k_1\oplus k_3]( t_1, t_3)$, then
  $s_1\mayto[ k_1]_1 t_1$ and $s_3\mayto[ k_3]_3 t_3$.  By $( s_1,
  s_2)\in R^1_{ \epsilon_1}\in R^1$, we have $s_2\mayto[ k_2]_2 t_2$
  with $\kdist( k_1, k_2)\le \epsilon_1$ and $( t_1, t_2)\in R^1_{
    \epsilon_1'}\in R^1$ for some $\epsilon_1'\le \lambda^{ -1}\big(
  \epsilon_1- \kdist( k_1, k_2)\big)$; similarly, $s_4\mayto[ k_4]_4 t_4$
  with $\kdist( k_3, k_4)\le \epsilon_2$ and $( t_3, t_4)\in R^2_{
    \epsilon_2'}\in R^2$ for some $\epsilon_2'\le \lambda^{ -1}\big(
  \epsilon_2- \kdist( k_3, k_4)\big)$.  Let $\epsilon'= \epsilon_1'+
  \epsilon_2'$, then the sum $k_2\oplus k_4$ is defined, and
  \begin{align*}
    \epsilon' &\le \lambda^{ -1}\big( \epsilon_1+ \epsilon_2-( \kdist( k_1,
    k_2)+ \kdist( k_3, k_4))\big) \\
    &\le \lambda^{ -1}\big( \epsilon- \kdist( k_1\oplus k_3, k_2\oplus
    k_4)\big)
  \end{align*}
  by Lemma~\ref{le:synch-cont}.  We have $( s_2, s_4)\mayto[ k_2\oplus
  k_4]( t_2, t_4)$, $\kdist( k_1\oplus k_3, k_2\oplus k_4)\le
  \epsilon_1+ \epsilon_2\le \epsilon$ again by
  Lemma~\ref{le:synch-cont}, and $\big(( t_1, t_3),( t_2, t_4)\big)\in
  R_{ \epsilon'}\in R$.  The reverse direction, starting with a
  transition $( s_2, s_4)\mustto[ k_2\oplus k_4]( t_2, t_4)$, is
  similar. \qed
\end{proof}

Again after a technical lemma, the next theorem expresses the fact that
quotient is a partial inverse to structural composition.  Intuitively,
the theorem shows that the quotient $S_1\bbslash S_2$ is maximal among
all WMTS $S_3$ with respect to any distance $S_2\| S_3\le_m^\epsilon
S_1$; note the special case of $S_3\le_m S_1\bbslash S_2$ if and only if
$S_2\| S_3\le_m S_1$.

\begin{lemma}
  \label{le:plusminus}
  If $k_1, k_2, k_3\in \K$ are such that $k_1\ominus k_2$ and $k_2\oplus
  k_3$ are defined, then $\kdist( k_3, k_1\ominus k_2)= \kdist(
  k_2\oplus k_3, k_1)$.
\end{lemma}

\begin{proof}
  We can write $k_i=\big( a,[ x_i, y_i]\big)$ for some $a\in \Sigma$.
  Then
  \begin{align*}
    \kdist( k_3, k_1\ominus k_2) &= \max\big(( x_1- x_2) - x_3,
    y_3 - ( y_1- y_2), 0\big) \\
    &= \left\{
      \begin{array}{cl}
        x_1- x_2- x_3 &\quad\text{if}\quad
        \begin{aligned}[t]
          x_1- x_2- x_3 &\ge 0, \\
          x_1- x_2- x_3 &\ge y_3- y_1+ y_2;
        \end{aligned} \\
        y_3- y_1+ y_2 &\quad\text{if}\quad
        \begin{aligned}[t]
          y_3- y_1+ y_2 &\ge 0, \\
          y_3- y_1+ y_2 &\ge x_1- x_2- x_3;
        \end{aligned} \\
        0 &\quad\text{if}\quad
        \begin{aligned}[t]
          x_1- x_2- x_3 &\le 0, \\
          y_3- y_1+ y_2 &\le 0.
        \end{aligned}
      \end{array}
    \right.
  \end{align*}
  Similarly,
  \begin{align*}
    \kdist( k_2\oplus k_3, k_1) &= \max\big( x_1 - ( x_2+ x_3),( y_2+
    y_3)- y_1, 0\big) \\
    &= \left\{
      \begin{array}{cl}
        x_1- x_2- x_3 &\quad\text{if}\quad
        \begin{aligned}[t]
          x_1- x_2- x_3 &\ge 0, \\
          x_1- x_2- x_3 &\ge y_2+ y_3- y_1 ;
        \end{aligned} \\
        y_2+ y_3- y_1 &\quad\text{if}\quad
        \begin{aligned}[t]
          y_2+ y_3- y_1 &\ge 0, \\
          y_2+ y_3- y_1 &\ge x_1- x_2- x_3;
        \end{aligned} \\
        0 &\quad\text{if}\quad
        \begin{aligned}[t]
          x_1- x_2- x_3 &\le 0, \\
          y_2+ y_3- y_1 &\le 0. \quad \qed
        \end{aligned}
      \end{array}
    \right.
  \end{align*}
\end{proof}

\begin{theorem}[Soundness and maximality of quotient]
  \label{th:soundmaxquot}
  Let $S_1$, $S_2$ and $S_3$ be locally consistent WMTS such that
  $S_2$ is deterministic and $S_1\bbslash S_2$ is defined.  If
  $\mdist( S_3, S_1\bbslash S_2)< \infty$, then $\mdist( S_3,
  S_1\bbslash S_2)= \mdist( S_2\| S_3, S_1)$.
\end{theorem}

\begin{proof}
  To avoid confusion, we write $\mmayto_\bbslash$ and
  $\mmustto_\bbslash$ for transitions in $S_1\bbslash S_2$ and
  $\mmayto_\|$ and $\mmustto_\|$ for transitions in $S_2\| S_3$.  The
  inequality $\mdist( S_3, S_1\bbslash S_2)\ge \mdist( S_2\| S_3, S_1)$
  is trivial if $\mdist( S_2\| S_3, S_1)= \infty$, so assume the
  opposite and let $R^1=\big\{ R^1_\epsilon\subseteq S_3\times\big(
  S_1\times S_2\cup\{ u\}\big)\bigmid \epsilon\ge 0\big\}$ be a
  witness for $\mdist( S_3, S_1\bbslash S_2)$.  Define
  $R^2_\epsilon=\big\{\big(( s_2, s_3), s_1\big)\bigmid\big(s_3,( s_1,
  s_2)\big)\in R^1_\epsilon\big\}\subseteq S_2\times S_3\times S_1$
  for all $\epsilon\ge 0$, and let $R^2=\{ R^2_\epsilon\mid
  \epsilon\ge 0\}$.  Certainly $\big(( s_2^0, s_3^0), s_1^0\big)\in
  R^2_{ \mdist( S_3, S_1\bbslash S_2)}\in R^2$, so let now $\big(( s_2,
  s_3), s_1\big)\in R^2_\epsilon\in R^2$ for some $\epsilon\ge 0$.

  Assume $( s_2, s_3)\mayto[ k_2\oplus k_3]_\|( t_2, t_3)$, then also
  $s_2\mayto[ k_2]_2 t_2$ and $s_3\mayto[ k_3]_3 t_3$.  We have $\big(
  s_3,( s_1, s_2)\big)\in R^1_\epsilon$, so there is $( s_1, s_2)\mayto[
  k_1\ominus k_2']_\bbslash( t_1, t_2')$ for which $\kdist( k_3,
  k_1\ominus k_2')= \kdist( k_2'\oplus k_3, k_1)\le \epsilon$ and such
  that $\big( t_3,( t_1, t_2')\big)\in R^1_{ \epsilon'}\in R^1$, hence
  $\big(( t_2', t_3), t_1\big)\in R^2_{ \epsilon'}\in R^2$, for some
  $\epsilon'\le \lambda^{ -1}\big( \epsilon- \kdist( k_2'\oplus k_3,
  k_1)\big)$.  By definition of quotient we must have $s_1\mayto[ k_1]_1
  t_1$ and $s_2\mayto[ k_2']_2 t_2'$, and by determinism of $S_2$,
  $k_2'= k_2$ and $t_2'= t_2$.

  Assume $s_1\mustto[ k_1]_1 t_1$.  We must have a transition
  $s_2\mustto[ k_2]_2 t_2$ for which $k_1\ominus k_2$ is defined.  Hence
  $( s_1, s_2)\mustto[ k_1\ominus k_2]_\bbslash( t_1, t_2)$.  This in
  turn implies that there is $s_3\mustto[ k_3]_3 t_3$ for which $\kdist(
  k_3, k_1\ominus k_2)= \kdist( k_2\oplus k_3, k_1)\le \epsilon$ and
  such that $\big( t_3,( t_1, t_2)\big)\in R^1_{ \epsilon'}\in R^1$,
  hence $\big(( t_2, t_3), t_1\big)\in R^2_{ \epsilon'}\in R^2$, for
  some $\epsilon'\le \lambda^{ -1}\big( \epsilon- \kdist( k_2\oplus k_3,
  k_1)\big)$, and by definition of parallel composition, $( s_2,
  s_3)\mustto[ k_2\oplus k_3]_\|( t_2, t_3)$.

  To show that $\mdist( S_3, S_1\bbslash S_2)\le \mdist( S_2\| S_3,
  S_1)$, let $R^2=\{ R^2_\epsilon\subseteq S_2\times S_3\times S_1\mid
  \epsilon\ge 0\}$ be a witness for $\mdist( S_2\| S_3, S_1)$, define
  $R^1_\epsilon=\big\{\big( s_3,( s_1, s_2)\big)\bigmid\big(( s_2,
  s_3), s_1\big)\in R^2_\epsilon\big\}\cup\big\{( s_3, u)\bigmid
  s_3\in S_3\big\}$ for all $\epsilon\ge 0$, and let $R^1=\{
  R^1_\epsilon\mid \epsilon\ge 0\}$, then $\big( s_3^0,( s_1^0,
  s_2^0)\big)\in R^1_{ \mdist( S_2\| S_3, S_1)}\in R^1$.

  For any $( s_3, u)\in R^1_\epsilon$ for some $\epsilon\ge 0$, any
  transition $s_3\mayto[ k_3]_3 t_3$ can be matched by $u\mayto[
  k_3]_\bbslash u$, and then $( t_3, u)\in R^1_0$.  Let now $\big(
  s_3,( s_1, s_2)\big)\in R^1_\epsilon$ for some $\epsilon\ge 0$, and
  assume $s_3\mayto[ k_3]_3 t_3$.  If $k_2\oplus k_3$ is undefined for
  all transitions $s_2\mayto[ k_2]_2 t_2$, then by definition $( s_1,
  s_2)\mayto[ k_3] u$, and again $( t_3, u)\in R^1_0$.  If there is a
  transition $s_2\mayto[ k_2]_2 t_2$ such that $k_2\oplus k_3$ is
  defined, then also $( s_2, s_3)\mayto[ k_2\oplus k_3]_\|( t_2,
  t_3)$.  Hence we have $s_1\mayto[ k_1]_1 t_1$ with $\kdist( k_2\oplus
  k_3, k_1)\le \epsilon$, implying that $( s_1, s_2)\mayto[ k_1\ominus
  k_2]_\bbslash( t_1, t_2)$. Hence $\kdist( k_3,
  k_1\ominus k_2)= \kdist( k_2\oplus k_3, k_1)\le \epsilon$.  Also, $\big((
  t_2, t_3), t_1\big)\in R^2_{ \epsilon'}\in R^2$, hence $\big( t_3,(
  t_1, t_2)\big)\in R^1_{ \epsilon'}\in R^1$, for some $\epsilon'\le
  \lambda^{ -1}\big( \epsilon- \kdist( k_3, k_1\ominus k_2)\big)$.

  Assume $( s_1, s_2)\mustto[ k_1\ominus k_2]_\bbslash( t_1, t_2)$,
  hence we have $s_1\mustto[ k_1]_1 t_1$ and $s_2\mustto[ k_2]_2
  t_2$.  It follows that $( s_2, s_3)\mustto[ k_2'\oplus k_3]_\|(
  t_2', t_3)$ with $\kdist( k_2'\oplus k_3, k_1)= \kdist( k_3,
  k_1\ominus k_2')\le \epsilon$ and such that $\big(( t_2', t_3),
  t_1\big)\in R^2_{ \epsilon'}\in R^2$, hence $\big( t_3,( t_1,
  t_2')\big)\in R^1_{ \epsilon'}\in R^1$, for some $\epsilon'\le
  \lambda^{ -1}\big( \epsilon- \kdist( k_3, k_1\ominus k_2')\big)$.  By
  definition of parallel composition we must have $s_2\mustto[ k_2']_2
  t_2'$ and $s_3\mustto[ k_3]_3 t_3$, and by determinism of $S_2$,
  $k_2'= k_2$ and $t_2'= t_2$. \qed
\end{proof}

\begin{figure}[tpb]
  \centering
  \subfigure[$S_1$]{
    \begin{tikzpicture}[->,>=stealth',shorten >=1pt,auto,node
      distance=2.0cm,initial text=,scale=0.8,transform shape]
      \tikzstyle{every node}=[font=\small] \tikzstyle{every
        state}=[fill=white,shape=circle,inner sep=.5mm,minimum size=6mm]
      \node[initial,state] (s1) at (0,0) {$s_1$};
      \node[state] (s2) at (2,0) {$t_1$};
      \path (s1) edge [densely dashed] node [above] {$a,[ 0, 0]$} (s2);
    \end{tikzpicture}}
  \hspace{2mm}
  \subfigure[$S_2$]{
    \begin{tikzpicture}[->,>=stealth',shorten >=1pt,auto,node
      distance=2.0cm,initial text=,scale=0.8,transform shape]
      \tikzstyle{every node}=[font=\small] \tikzstyle{every
        state}=[fill=white,shape=circle,inner sep=.5mm,minimum size=6mm]
      \node[initial,state] (s1) at (0,0) {$s_2$};
      \node[state] (s2) at (2,0) {$t_2$};
      \path (s1) edge [densely dashed] node [above] {$a,[ 0, 1]$} (s2);
    \end{tikzpicture}}
  \hspace{2mm}
  \subfigure[$S_3$]{
    \begin{tikzpicture}[->,>=stealth',shorten >=1pt,auto,node
      distance=2.0cm,initial text=,scale=0.8,transform shape]
      \tikzstyle{every node}=[font=\small] \tikzstyle{every
        state}=[fill=white,shape=circle,inner sep=.5mm,minimum size=6mm]
      \node[initial,state] (s1) at (0,0) {$s_3$};
      \node[state] (s2) at (2,0) {$t_3$};
      \path (s1) edge [densely dashed] node [above] {$a,[ 0, 0]$} (s2);
    \end{tikzpicture}}
  \hspace{2mm}
  \subfigure[$S_2\| S_3$]{
    \begin{tikzpicture}[->,>=stealth',shorten >=1pt,auto,node
      distance=2.0cm,initial text=,scale=0.8,transform shape]
      \tikzstyle{every node}=[font=\small] \tikzstyle{every
        state}=[fill=white,shape=circle,inner sep=.5mm,minimum size=6mm]
      \node[initial,state,shape=rectangle,rounded corners] (s1) at
      (0,0) {$( s_2, s_3)$}; 
      \node[state,shape=rectangle,rounded corners] (s2) at (3,0)
      {$( t_2, t_3)$};
      \path (s1) edge [densely dashed] node [above] {$a,[ 0, 1]$} (s2);
    \end{tikzpicture}}
  \hspace{2mm}
  \subfigure[$S_1\bbslash S_2$]{
    \begin{tikzpicture}[->,>=stealth',shorten >=1pt,auto,node
      distance=2.0cm,initial text=,scale=0.8,transform shape]
      \tikzstyle{every node}=[font=\small] \tikzstyle{every
        state}=[fill=white,shape=circle,inner sep=.5mm,minimum size=6mm]
      \node[initial,state,shape=rectangle,rounded corners] (s1) at
      (0,0) {$( s_1, s_2)$};
    \end{tikzpicture}}
  \caption{WMTS for which $\mdist( S_2\| S_3, S_1)\ne \mdist( S_3,
    S_1\bbslash S_2)= \infty$.}
  \label{fi:quotient-counterex}
\end{figure}

The example depicted in Figure~\ref{fi:quotient-counterex} shows that
the condition $\mdist( S_3, S_1\bbslash S_2)< \infty$ in
Theorem~\ref{th:soundmaxquot} is necessary.  Here $\mdist( S_2\| S_3,
S_1)= 1$, but $\mdist( S_3, S_1\bbslash S_2)= \infty$ because of
inconsistency between the transitions $s_1\mayto[{ a,[ 0, 0]}]_1 t_1$
and $s_2\mayto[{ a,[ 0, 1]}]_2 t_2$ for which $k_1\ominus k_2$ is defined.

\medskip%
As a practical application, we notice that \emph{relaxation} as defined
in Section~\ref{se:relax} can be useful when computing quotients.  The
quotient construction in Definition~\ref{de:comp-quot} introduces
inconsistent states (which afterwards are pruned) whenever there is a
\must transition $s_1\mustto[ k_1]_1 s_1'$ such that $k_1\ominus k_2$ is
undefined for all transitions $s_2\mustto[ k_2]_2 s_2'$.  Looking at the
definition of $\ominus$, we see that this is the case if $k_1=( a_1,[
x_1, y_1])$ and $k_2=( a_2,[ x_2, y_2])$ are such that $a_1\ne a_2$ or
$x_1- x_2> y_1- y_2$.  In the first case, the inconsistency is of a
\emph{structural} nature and cannot be dealt with; but in the second
case, it may be avoided by \emph{enlarging} $k_1$: decreasing $x_1$ or
increasing $y_1$ so that now, $x_1- x_2\le y_1- y_2$.

Enlarging quantitative constraints is exactly the intuition of
relaxation, thus in practical cases where we get a quotient $S_1\bbslash
S_2$ which is ``too inconsistent'', we may be able to solve this problem
by constructing a suitable $\epsilon$-relaxation $S_1'$ of $S_1$.
Theorems~\ref{th:indepimp} and~\ref{th:soundmaxquot} can then be used to
ensure that also $S_1'\bbslash S_2$ is a relaxation of $S_1\bbslash
S_2$.

\section{Logical Characterizations}
\label{se:logics}

We now turn our attention to showing that quantitative refinement admits
a logical characterization. Our results extend the logical
characterization of modal transition systems in~\cite{Larsen89}, by
abandoning the usual Boolean interpretation of logical satisfaction, as
we did for refinement, and instead interpreting each formula as a map
assigning to states a real-valued number denoting the relationship
between the property and the state.  The logic $\mathcal L$ is the
smallest set of expressions generated by the following abstract syntax:
\begin{equation*}
  \phi,\phi_1,\phi_2 := \true \mid \false
  \mid  \langle \ell\rangle \phi \mid [ \ell] \phi \mid
  \phi_1 \wedge \phi_2
  \mid \phi_1 \vee \phi_2 \qquad (\ell\in \K)
\end{equation*}
As usual, when $\ell = (a,[x_1,x_2])$, writing $\langle \ell \rangle
\phi$ means that we insist on implementations exhibiting a transition
which reaches a state having property $\phi$ and is labeled by $a$ and
an integer $x$ for which $x_1\le x\le x_2$. Dually, $[\ell]\phi$
restricts the set of implementations to those where every transition
labeled with $a$ and an integer in $[x_1,x_2]$ reaches a state with
property $\phi$.

With this standard (informal) interpretation of logical specifications,
implementations which come close to matching the specification are
rejected just as much as the truly wrong implementations. Analog to our
refinement distance, a quantitative interpretation provides us with
continuous judgments on the relationship between a specification $S$ or
implementation $I$ and a logical specification $\phi$. Defining the
semantics of formulae as a map from states to reals, the value of any
$\phi$ for the initial state of implementations determines an order on
the applicability of the implementations for the given specification.
The semantics of a formula $\phi\in \mathcal L$ is a mapping $\sem \phi:
S\to \Realnn \cup \{\infty\}$ given inductively, again relative to the
discounting factor $\lambda$ with $0 < \lambda <1$, as follows:
\begin{gather*}
  \begin{aligned}
    \sem \true s &= 0 &\qquad\qquad \sem \false s &= \infty \\
    \sem{( \phi_1\wedge \phi_2)} s &= \max( \sem{ \phi_1} s, \sem \phi_2
    s) &\qquad \sem{( \phi_1\vee \phi_2)} s &= \min( \sem{ \phi_1} s,
    \sem{ \phi_2} s)
  \end{aligned}
  \\
  \begin{aligned}
    \sem{ \langle \ell\rangle \phi} s &= \inf\{ \kdist( k, \ell) +
    \lambda \sem \phi
    t)\mid s\mustto[ k] t, \kdist( k, \ell)\ne \infty\} \\
    \sem{ [ \ell] \phi} s &= \sup\{ \kdist( k, \ell) + \lambda \sem
    \phi t \mid s\mayto[ k] t, \kdist( k, \ell)\ne \infty\}
  \end{aligned}
\end{gather*}

Intuitively, $\sem{[\ell]\phi}s$ takes the value of the supremum over
all outgoing $s\mayto[a,x] t$ transitions and the respective match with
$x\in [x_1,x_2]$ plus the discounted value of the property $\phi$ for
$t$. Clearly if $\sem{[\ell]\phi}s = 0$ then every $s\mayto[a,x] t$
satisfies the property exactly, recovering the standard interpretation.
Notice that by evaluating a logical specification $\phi$ for a WMTS
specification $S$, we get a measure on the set of implementations of $S$
which are not shared by the specification $\phi$.  The value is $0$ if
and only if there is a thorough refinement from $S$ to $\phi$, \ie~if
and only if any implementation of $S$ satisfies $\phi$.

For a SMTS $S$ we write $\sem \phi S= \sem \phi s_0$.  The first theorem
below expresses the fact that $\mathcal L$ is \emph{quantitatively
  sound} for refinement distance, \ie~the value of a formula in a
specification is bounded by its value in any other specification
together with their distance.  Note the special case that $S\le_m T$
implies $\sem \phi S\le \sem \phi T$.

\begin{theorem}
  \label{th:l-sound}
  For all $\phi\in \mathcal L$ and WMTS $S$, $T$, $\sem \phi
  S \le \sem \phi T + d_m( S, T)$.
\end{theorem}

\begin{proof}
  By standard structural induction in $\phi$.  The claim obviously holds
  for $\phi= \true$ and $\phi= \false$.

  For $\phi= \phi_1\land \phi_2$, the induction hypothesis that $\sem{
    \phi_i} s_1\le \sem{ \phi_i} s_2 + d_m( s_1, s_2)$ for $i= 1, 2$
  implies that also $\max( \sem{ \phi_1} s_1, \sem{ \phi_2} s_1)\le
  \max( \sem{ \phi_1} s_2, \sem{ \phi_2} s_2) + d_m( s_1, s_2)$.
  Similarly for $\phi= \phi_1\vee \phi_2$.

  For the case $\phi=\langle \ell\rangle \phi'$, if $d_m( s_1, s_2)=
  \infty$ or if there are no transitions $s_2\mustto$ the claim is
  trivial.  Let thus $s_2\mustto[ k_2] t_2$, then there exist
  $s_1\mustto[ k_1] t_1$ with $\kdist( k_1, k_2) + \lambda d_m( t_1,
  t_2) \le d_m( s_1, s_2)$ (by definition of $d_m$).  

  Then $\kdist( k_1, \ell) + \lambda \sem{ \phi'} t_1 \le (\kdist(
  k_1, k_2) + \lambda d_m( t_1, t_2)) + (\kdist( k_2, \ell) + \lambda
  \sem{ \phi'} t_2)$ by induction hypothesis and the triangle
  inequality for $\kdist$, hence $\kdist( k_1, \ell) + \lambda \sem{
    \phi'} t_1 \le d_m( s_1, s_2) + \kdist( k_2, \ell)+\lambda \sem{
    \phi'} t_2$.  As $s_2\mustto[ k_2] t_2$ was arbitrary, this
  entails $\inf\{ \kdist( k_1, \ell) + \lambda \sem{ \phi'} t_1 \mid
  s_1\mustto[ k_1] t_1\}\le \inf\{ \kdist( k_2, \ell) + \lambda \sem{
    \phi'} t_2\mid s_1\mustto[ k_2] t_2\} + d_m( s_1, s_2)$, which was
  to be shown.

  For the case of $\phi=[ \ell] \phi'$ the proof is similar: We have
  nothing to prove if $d_m( s_1, s_2)= \infty$ or if there are no
  transitions $s_1\mayto[ k_1] t_1$ with $\kdist( k_1, \ell) \ne
  \infty$, so assume there is such a transition.  Then we also have
  $s_2\mayto[ k_2] t_2$ with $(\kdist( k_1, k_2) + \lambda d_m( t_1,
  t_2))\le d_m( s_1, s_2)$, and $\kdist( k_1, \ell)+\lambda \sem{
    \phi'} t_1 \le (\kdist( k_1, k_2)+\lambda d_m( t_1, t_2)) +
  \kdist( k_2, \ell) +\lambda \sem{ \phi'} t_2 \le d_m( s_1, s_2) +
  \kdist( k_2, \ell)+\lambda \sem{ \phi'} t_2$.  \qed
\end{proof}

The next theorem shows that the disjunction-free fragment of $\mathcal
L$ is also \emph{quantitatively implementation complete}, \ie~the value
of any disjunction-free formula in a specification $S$ is bounded above
by its value in any implementation of $S$.  Note that
disjunction-freeness is a common assumption in this context,
\cf~\cite{Larsen89,DBLP:conf/atva/BenesCK11}.

\begin{theorem}
  \label{th:l-complete}
  For all disjunction-free $\phi\in \mathcal L$ and locally consistent
  and compactly branching WMTS $S$, we have $\sem \phi S= \sup_{ I\in \llbracket
    S\rrbracket} \sem \phi I$.
\end{theorem}

\begin{proof}
  Since $d_m(I,S) = 0$ for all $I\in\sem{S}$, Theorem~\ref{th:l-sound}
  entails $\sem{\phi}I \le \sem{\phi}S$, hence also
  $\sup_{I\in\sem{S}}\sem{\phi}I \le \sem{\phi}S$. To show that
  $\sem{\phi}S \le \sup_{I\in\sem{S}}\sem{\phi}I$ we use structural
  induction on $\phi$.  If $\phi= \true$, both sides are $0$, and if
  $\phi= \false$, both sides are $\infty$, so the induction base is
  clear.

  The case $\phi= \phi_1\wedge \phi_2$ is also clear: By hypothesis,
  $\sem{ \phi_1}S \le \sup_{ I\in \llbracket S\rrbracket} \sem{
    \phi_1} I$ and similarly for $\phi_2$, hence
  \begin{align*}
    \sem \phi S= \max( \sem{ \phi_1} S, \sem{ \phi_2} S) &\le
    \max( \sup_{ I\in \llbracket S\rrbracket} \sem{ \phi_1} I, \sup_{
      I\in \llbracket S\rrbracket} \sem{ \phi_2} I) \\ &= \sup_{ I\in
      \llbracket S\rrbracket} \max( \sem{ \phi_1} I, \sem{ \phi_2} I).
  \end{align*}

  For the case $\phi= \langle \ell\rangle \phi'$, we are done if $\sem
  \phi S= 0$.  Otherwise, to conclude that $\sup_{I \in\sem{S}}
  \sem{\langle \ell\rangle \phi'} I \ge \sem{\langle \ell\rangle
    \phi'}S$ we expose an $I\in \llbracket S\rrbracket$ for which
  $\alpha < \sem \phi I$ for any $\alpha < \sem \phi S$. For a fixed
  $\alpha < \sem \phi S$, start by letting $I=\{ i_0\}$ and $\mmustto_I=
  \emptyset$.

  Now for each transition $s_0\mustto[ k]_S t$ we have $\alpha < \kdist(
  k, \ell)+ \lambda\sem{ \phi'} t$, so (assuming for the moment that
  $\sem{ \phi'} t\ne 0$) by the density of the reals, there is a number
  $\alpha_k' < \sem{ \phi'} t$ for which $\alpha < \kdist( k,
  \ell)+\lambda \alpha_k'$.  By induction hypothesis, the sub-formula
  $\phi'$ satisfies $\sup_{J\in \sem{S'}} \sem{\phi'} J = \sem{\phi'}
  S'$ for any $S'$, specifically when $S'=( t, S)$ is taken as $S$ with
  initial state replaced by $t$. Therefore, and as $\alpha_k' <
  \sem{\phi'} t$, there exists a $J\in \sem{(t,S)}$ with $\alpha_k' <
  \sem{\phi'}J$.  Now let $n\in \Imp \K$ with $n\sqsubseteq k$ be such that
  $\kdist( n, \ell) + \lambda \sem{ \phi'} J= \kdist( k, \ell)+\lambda
  \sem{ \phi'} J$, and add $J$ together with a transition $i_0\mustto[
  n]_I j_0$ to $I$.

  In case $\sem{ \phi'} t= 0$, we have $J\in \llbracket t, S\rrbracket$
  with $\sem{ \phi'} J= 0$, and we can add $J$ together with a
  transition $i_0\mustto[ n]_I j_0$ to $I$ as above.

  For the so-constructed implementation $I$ we have
  \begin{align}
    \sem \phi I &= \inf\{ \kdist( m, \ell)+\lambda \sem{ \phi'} j\mid
    i_0\mustto[
    m]_I j\} \notag\\
    &= \inf\{ \kdist( k, \ell)+\lambda \sem{ \phi'} J\mid s_0\mustto[
    k]_S t, J\in \llbracket t, S\rrbracket, \sem{ \phi'} t= \infty
    \text{ or }
    \alpha_k' < \sem{ \phi'} J\} \notag\\
    &> \inf(\{ \kdist( k, \ell)+\lambda \alpha_k' \mid s_0\mustto[ k]_S
    t\}\cup\{ \kdist( k, \ell)+\lambda \sem{ \phi'} t \})\ge
    \alpha, \label{eq:fibragen}
  \end{align}
  the strict inequality in~\eqref{eq:fibragen} because $S$ is compactly
  branching.

  For the case $\phi=[ \ell] \phi'$, let again $\alpha < \sem \phi S$,
  and let $I\in \llbracket S\rrbracket$ be any implementation.  If
  $\kdist( k, \ell)+\lambda \sem{ \phi'} t= \infty$ for all $s_0\mayto[
  k]_S t$, then $\sem \phi S= \sup \emptyset= 0$ and we are done.
  Otherwise let $s_0\mayto[ k]_S t$ be such that $\sem \phi S= \kdist(
  k, \ell)+\lambda \sem{ \phi'} t$, which exists because $S$ is
  compactly branching.  Then $\alpha < \kdist( k, \ell)+\lambda \sem{
    \phi'} t$, so (assuming that $\sem{ \phi'} t \ne 0$) we have
  $\alpha_k'< \sem{ \phi'} t$ with $\kdist( k, \ell)+\lambda \alpha_k' >
  \alpha$.

  Let $J\in \llbracket t, S\rrbracket$ such that $\alpha_k'< \sem{
    \phi'} J$, let $n\in \Imp \Spec$ with $n\sqsubseteq k$ be such that
  $\kdist( n, \ell)+\lambda \sem{ \phi'} J = \kdist( k, \ell)+\lambda
  \sem{ \phi'} J$, and add $J$ together with a transition $i_0\mustto[
  n]_I j_0$ to $I$.  Then
  \begin{align*}
    \sem \phi I &= \sup\{ \kdist( m, \ell)+\lambda \sem{ \phi'} n \mid
    i_0\mustto[
    m]_I j\} \\
    &\ge \kdist( n, \ell)+\lambda \sem{ \phi'} J = \kdist( k,
    \ell)+\lambda \sem{ \phi'} J \ge F( k, \ell,
    \alpha_k') > \alpha.
  \end{align*}
  In case $\sem{ \phi'} t= 0$ instead, we again take some $J\in
  \llbracket t, S\rrbracket$, and then $\sem \phi I\ge \kdist( k,
  \ell)+\lambda \sem{ \phi'} t > \alpha$. \qed
\end{proof}

Other notions of completeness (see
e.g.~\cite{journals/mscs/BauerJLLS11}) are subject of future work.

\section{Conclusion and Further Work}

We have shown in this paper that within the quantitative specification
framework of weighted modal transition systems, refinement and
implementation distances provide a useful tool for robust compositional
reasoning.  Note that these distances permit us not only to reason about
differences between implementations and from implementations to
specifications, but they also provide a means by which we can compare
specifications directly at the abstract level.

We have shown that for some of the ingredients of our specification
theory, namely structural composition and quotient, our formalism is a
conservative extension of the standard Boolean notions.  We have also
noted however, that for determinization and logical conjunction, the
properties of the Boolean notions are not preserved, %
and that this is a fundamental limitation of any reasonable quantitative
specification theory.  The precise practical implications of this for
the applicability of our quantitative specification framework, and
perhaps how to circumvent these limitations, are subject to future work.


\end{document}